\documentclass[11pt]{article}

\usepackage{graphicx}
\usepackage[margin=0.5in]{geometry}
\usepackage{bm}
\usepackage{amsthm}
\usepackage{amscd}
\usepackage{amsbsy}
\usepackage{appendix}
\usepackage{amsmath}
\usepackage{amsfonts}
\usepackage{amssymb}
\usepackage{calligra}
\usepackage{xcolor}
\usepackage{float}
\usepackage[T1]{fontenc}
\usepackage{multirow}
\usepackage{mathrsfs}
\usepackage{mathtools}
\usepackage{epstopdf}
\usepackage{psfrag}
\usepackage{subfigure}
\usepackage{booktabs}
\usepackage{listings}
\usepackage{longtable}
\usepackage{latexsym}
\usepackage{arydshln}
\usepackage{enumerate}
\usepackage{epsfig}
\usepackage{cite}
\usepackage{verbatim}
\usepackage{overpic}
\usepackage{abstract}
\usepackage{extarrows}
\usepackage[pagewise]{lineno}

\usepackage[noend]{algpseudocode}
\usepackage{algorithmicx,algorithm}

\usepackage[figuresleft]{rotating}

\usepackage{caption}
\usepackage{footnote}
\makesavenoteenv{tabular}

\allowdisplaybreaks[4]

\date{}
\title{Algorithmic Detection of Jacobi Stability for Systems of Second Order Differential Equations}

\author{Christian G. B\"ohmer$^{\text{a}}$, Bo Huang$^{\text{b},}$\footnote{Corresponding author.} , Dongming Wang$^{\text{c,d}}$, and Xinyu Wang$^{\text{b}}$ \vspace{0.2cm} \\
        \it\footnotesize $^{\text{a}}$Department of Mathematics, University College London, WC1E 6BT London, United Kingdom\\
        \it\footnotesize c.boehmer@ucl.ac.uk\\
	\it\footnotesize $^{\text{b}}$LMIB -- School of Mathematical Sciences,
Beihang University, Beijing 100191, China \\
	\it\footnotesize bohuang0407@buaa.edu.cn, xinyuw@buaa.edu.cn\\
	\it\footnotesize $^{\text{c}}$LMIB -- Institute of Artificial Intelligence, Beihang University, Beijing 100191, China\\
    \it\footnotesize $^{\text{d}}$Centre National de la Recherche Scientifique, 75794 Paris Cedex 16, France\\
	\it\footnotesize Dongming.Wang@buaa.edu.cn
	}

\newtheorem {theorem*}{Theorem}
\newtheorem{theorem} {Theorem}

\newtheorem{definition}{Definition}

\newtheorem{lemma}{Lemma}

\newtheorem{remark}{Remark}

\newtheorem{open problem} {Open problem}

\numberwithin{equation}{section}

\usepackage[colorlinks,
            CJKbookmarks=true,
            linkcolor=blue,
            anchorcolor=red,
             urlcolor=blue,
            citecolor=blue
            ]{hyperref}

\begin{document}

\maketitle

\noindent {\bf Abstract.} {This paper introduces an algorithmic approach to the analysis of Jacobi stability of systems of second order ordinary differential equations (ODEs) via the Kosambi--Cartan--Chern (KCC) theory. We develop an efficient symbolic program using Maple for computing the second KCC invariant for systems of second order ODEs in arbitrary dimension. The program allows us to systematically analyze Jacobi stability of a system of second order ODEs by means of real solving and solution classification using symbolic computation. The effectiveness of the proposed approach is illustrated by a model of wound strings, a two-dimensional airfoil model with cubic nonlinearity in supersonic flow and a 3-DOF tractor seat-operator model. The computational results on Jacobi stability of these models are further verified by numerical simulations. Moreover, our algorithmic approach allows us to detect hand-guided computation errors in published papers.}

\smallskip

\noindent {\bf Math Subject Classification (2020).} 34C07; 68W30.

\smallskip


\smallskip


\section{Introduction}

Second order ordinary differential equations (ODEs) are important mathematical objects because they have a large variety of applications in different domains of mathematics, science and engineering (see \cite{AIM1993}). The modern geometry of a second order ODE was initiated in the 1920s by Synge \cite{Synge1926}, Knebelman \cite{MSK1929}, Douglas \cite{Douglas1928} and the geometric invariants of a second order ODE were obtained in the 1930s by Kosambi \cite{Kosambi1933}, Cartan \cite{Cartan1933}, and Chern \cite{Chern1939}. In their papers they have considered a system of second order ODEs
\begin{equation}\label{eq1-0}
\frac{d^2x_i}{d t^2}+2G^i\left(t,x,\frac{d x}{d t}\right)=0,\quad i\in\{1,2,\ldots,n\}
\end{equation}
where $(t,x_i)$ are the local coordinates on a real $(n+1)$-dimensional fibred manifold $\pi:\mathcal{M}\rightarrow\mathbb{R}$. The main problem they were looking for was to find the geometric properties we can associate to the system \eqref{eq1-0} that are invariant under the following groups of transformations:
\begin{equation}\label{eq1-1}
(A)
\left\{\begin{array}{l}
\tilde{t}=t,\\
\tilde{x}_i=\tilde{x}_i(x_j)\\
\end{array}\right. \quad \text{or} \quad (B)
\left\{\begin{array}{l}
\tilde{t}=t,\\
\tilde{x}_i=\tilde{x}_i(t,x_j).\\
\end{array}\right.
\end{equation}
As far as we know, the problem is not completely solved despite significant progress in the subject, see \cite{CMS1996,Olga1997,Brad1999,BChernS2000,Shen2001} and references therein. 

The geometry of the system \eqref{eq1-0} under the action of the group of transformations (A) is called the Kosambi--Cartan--Chern (KCC) theory of type (A), while the geometry of the system \eqref{eq1-0} under the group of transformations (B) is called the KCC theory of type (B). The KCC theory is established from the variational equations for the deviation of the whole trajectory to nearby ones, based on the existence of a one-to-one correspondence between a second order ODE and the geodesic equations in an associated Finsler space (see \cite{BHS2012} for more details). By associating a nonlinear connection and a Berwald-type connection (these are discussed in detail in Section \ref{sect2}) to a system of second order ODE, five geometrical invariants are obtained, with the second invariant giving the Jacobi stability of the system. The KCC theory has been applied for the study of different physical, biochemical or engineering systems (see \cite{PLA2003,SVS12005,SVS22005,BHS2012,HHLY2015,HPPS2016,GY2017} and references therein).

The main object of this work is a system of $n$ homogeneous second order ODEs, whose coefficients do not depend explicitly on time,
\begin{equation}\label{eq1-2}
\begin{split}
\frac{d^2x_i}{d t^2}+2G^i\left(\mu;x,\frac{d x}{d t}\right)=0,\quad i\in\{1,2,\ldots,n\}
\end{split}
\end{equation}
where $x=(x_1,\ldots,x_n)$ are variables, $\mu=(\mu_1,\ldots,\mu_p)$ are real parameters, and $G^i$ are rational functions in $\mathbb{R}(x,dx/dt)$. Our study is focused on describing the properties of solution trajectories of the system \eqref{eq1-2} by using the KCC theory of type (A). More concretely, the questions we answer include the following:

\begin{enumerate}
\item [1.]   Under  what conditions does a system of the form \eqref{eq1-2} have a prescribed number of Jacobi stable fixed points?
\item [2.]  Given a system of the form \eqref{eq1-2}: what can be said about the focusing/dispersing behavior of the trajectories near the fixed points?
\end{enumerate}

The investigations in the present paper are based mainly on the KCC theory and some algebraic methods with exact symbolic computation. Our answer to the first question is an effective algorithmic approach using polynomial-algebra methods to solve semi-algebraic systems, and to the second, a method based on the KCC theory for performing numerical simulations of the obtained deviation equations (combined with the solution to the first). We remark that the qualitative analysis of differential equations in practice usually involves heavy symbolic computation. In the last three decades, remarkable progress has been made on the research and development of symbolic and algebraic algorithms and software tools, bringing the classical qualitative theory of differential equations to computational approaches for symbolic analysis of the qualitative behaviors of diverse differential equations (see \cite{HLS97,WDMXBC2005,HTX15,dw91,EKW00,swek-mcs2009,vd09,adrt2016,acj17,HGR10,CCMYZ13,bhlr2015,bdetal:2020}, the survey article \cite{HNW2022}, and references therein). 

The Jacobi stability of a system of second order ODEs can be regarded as the robustness of the system to small perturbations of the whole trajectory \cite{SVS22005}. The system of differential equations describing the deviations of the whole trajectory of a system of second order ODEs, with respect to perturbed trajectories, are introduced geometrically in KCC theory via the second KCC invariant (or the deviation curvature tensor). In our recent work \cite{ISSAC2024HWY}, we provide a computational approach to the analysis of the Jacobi stability for systems of first order differential equations. The information of the Jacobian matrix of the first order equations can be used to deduce the second KCC invariant (see \cite[Lemma 3.1]{ISSAC2024HWY}). However, in practice, the evaluation of the second KCC invariant of a system of second order ODEs is generally a computationally challenging problem. Moreover, the computational complexity grows very fast when the dimension of the required systems increases. Thus, traditional hand-guided computations may lead to errors in dealing with large symbolic expressions. The aim of this paper is twofold: first, we develop an efficient symbolic program using Maple for computing the second KCC invariant for systems of the form \eqref{eq1-2} in arbitrary dimension. Second, we use symbolic computation methods (combined with the program) to provide a systematical and algorithmic approach for analyzing Jacobi stability of a system of second order ODEs \eqref{eq1-2}.

The rest of the paper is structured as follows. We explain the basic concepts and results of the KCC theory in Section \ref{sect2}. Section \ref{sect-main} contains the main results, such as Theorem \ref{kcc-deviation}, which gives the parametric formula of the deviation equations associated to system \eqref{eq1-2}; and Theorem \ref{semi-kcc}, which provides a computational approach to the conditions for a system of the form \eqref{eq1-2} to have given numbers of Jacobi stable fixed points. In Section \ref{sect4}, we demonstrate the effectiveness of our algorithmic approach using results obtained for a model of wound strings, a two-dimensional airfoil model with cubic nonlinearity in supersonic flow and a 3-DOF tractor seat-operator model. The paper ends with some discussions in Section \ref{sect5}.

The Maple program developed in this paper can be found at{\small{
{\textcolor{blue}{\underline{\url{https://github.com/Bo-Math/KCC-Jacobi}}}}}}. The numerical simulations on the focusing/dispersing behaviors of the three models together with dynamic behaviors of the corresponding deviation equations near the fixed points are given in the appendices. 


\section{KCC Theory and Jacobi Stability}\label{sect2}

In this section, we briefly recall the basics of KCC theory to be used in the sequel. For more details about the KCC theory, the reader may consult the expository article \cite{BHS2012} and other recent work in \cite{PLA2003,SVS12005,SVS22005,HPPS2016}.

Let $\mathcal{M}$ be a real, smooth $n$-dimensional manifold and let $\mathcal{TM}$ be its tangent bundle. In applications, $\mathcal{M}=\mathbb{R}^n$ or $\mathcal{M}\subset\mathbb{R}^n$ is an open set. Let $\boldsymbol{p}=(x,y)$ be a point in $\mathcal{TM}$, where $x=(x_1,\ldots,x_n)$ and $y=(y_1,\ldots,y_n)$, with $y_i=dx_i/dt$, $i=1,\ldots,n$.

Consider the following system of second order ODEs of the form \eqref{eq1-2}
\begin{equation}\label{kcc-4}
\begin{split}
\frac{d^2x_i}{d t^2}+2G^i\left(\mu;x,y\right)=0,\quad i\in\{1,2,\ldots,n\},
\end{split}
\end{equation}
where $\mu$ are real parameters, $G^i(\mu;x,y)$ are $\mathcal{C}^{\infty}$ functions locally defined in a coordinates chart on $\mathcal{TM}$, i.e., an open set around a point
$(x_0,y_0)\in\mathcal{TM}$ with initial conditions. The system \eqref{kcc-4} is inspired by the Euler--Lagrange equations of classical dynamics
\begin{equation}\label{kcc-3}
\begin{split}
\frac{d}{dt}\frac{\partial L}{\partial y_i}-\frac{\partial L}{\partial x_i}=F_i,\quad i=1,2,\ldots,n,
\end{split}
\end{equation}
where $L$ is the Lagrangian of $\mathcal{M}$, and $F_i$ are external forces. 

In order to find the basic differential invariants of system \eqref{kcc-4} under the action of the group
of transformations (A) given in \eqref{eq1-1}, we define the KCC-covariant differential of a vector field $\xi_i(\boldsymbol{x})$ on an open subset $\Omega\subset\mathbb{R}^n\times\mathbb{R}^n$ by
\begin{equation}\label{kcc-5}
\begin{split}
\frac{D\xi_i}{dt}=\frac{d\xi_i}{dt}+N_j^i\xi_j,
\end{split}
\end{equation}
where $N_j^i=\partial G^i/\partial y_j$ are the coefficients of the nonlinear connection. The original idea of this approach belongs to Kosambi \cite{Kosambi1933}, Cartan \cite{Cartan1933} (who corrected Kosambi's work), and Chern (for the most general version) \cite{Chern1939}. The Einstein summation convention will be used throughout.

Using \eqref{kcc-5}, the system \eqref{kcc-4} becomes
\begin{equation}\label{kcc-6}
\begin{split}
\frac{Dy_i}{dt}=N_j^iy_j-2G^i=-\epsilon_i,\nonumber
\end{split}
\end{equation}
where $\epsilon_i$ is a contravariant vector field on $\Omega$ and is called the first KCC invariant. It is interpreted as an external force \cite{Antonelli2000}.


We now vary the trajectories $x_i(t)$ of the system \eqref{kcc-4} into nearby ones according to
\begin{equation}\label{kcc-7}
\begin{split}
\tilde{x}_i(t)=x_i(t)+\eta\xi_i(t),
\end{split}
\end{equation}
where $|\eta|$ is a small parameter, and $\xi_i(t)$ are the components of a contravariant vector field defined along the path $x_i(t)$. Substituting \eqref{kcc-7} into \eqref{kcc-4} and taking the limit $\eta\rightarrow0$, we obtain the so-called deviation, or Jacobi, equations:
\begin{equation}\label{kcc-8}
\begin{split}
\frac{d^2\xi_i}{dt^2}+2N_j^i\frac{d\xi_j}{dt}+2\frac{\partial G^i}{\partial x_j}\xi_j=0.
\end{split}
\end{equation}

By using the KCC-covariant derivative, defined by \eqref{kcc-5}, we can rewrite equation \eqref{kcc-8} in a covariant form as
\begin{equation}\label{kcc-9}
\begin{split}
\frac{D^2\xi_i}{dt^2}=P_j^i\xi_j,
\end{split}
\end{equation}
where we have denoted
\begin{equation}\label{kcc-10}
\begin{split}
P_j^i=-2\frac{\partial G^i}{\partial x_j}-2G^{\ell}G^i_{j\ell}+y_{\ell}\frac{\partial N_j^i}{\partial x_{\ell}}+N_{\ell}^iN_j^{\ell},
\end{split}
\end{equation}
and introduced the Berwald connection $G^i_{j\ell}$, defined as
\begin{equation}\label{kcc-11}
\begin{split}
G^i_{j\ell}\equiv\frac{\partial N_j^i}{\partial y_{\ell}}.
\end{split}
\end{equation}
$P^i_j$ is called the second KCC-invariant or the deviation curvature tensor. 

Note that \eqref{kcc-9} is the Jacobi field equation whenever system \eqref{kcc-4} is formed with the geodesic equations in either Finsler or Riemannian geometry. The notion of Jacobi stability for systems of the form \eqref{kcc-4} is thus a generalization of that for the geodesics of a Riemannian or Finsler manifold. This justifies the usage of the term Jacobi stability for KCC theory.

In the KCC theory one can also introduce the third, fourth and fifth invariants of the system \eqref{kcc-4}. These invariants are called respectively the torsion tensor, the Riemann--Christoffel curvature tensor, and the Douglas tensor, which are defined as
\begin{equation}\label{kcc-12}
\begin{split}
P^i_{jk}=\frac{1}{3}\left(\frac{\partial P^i_j}{\partial y_k}-\frac{\partial P^i_k}{\partial y_j}\right),\quad P^i_{ik\ell}=\frac{\partial P^i_{jk}}{\partial y_{\ell}},\quad D^i_{jk\ell}=\frac{\partial G^i_{jk}}{\partial y_{\ell}}.\nonumber
\end{split}
\end{equation}
These tensors always exist in a Berwald space. In the KCC theory they describe the geometrical properties and interpretation of a system of second order ODEs (see \cite{PLA2003,SVS12005,SVS22005} and references therein).

For many physical, chemical or biological applications one is interested in the behaviors of the trajectories of the system \eqref{kcc-4} in a vicinity of a point $x_i(t_0)$. For simplicity in the following we take $t_0=0$. We consider the trajectories $x_i=x_i(t)$ of \eqref{kcc-4} as curves in the Euclidean space $(\mathbb{R}^n,\langle\cdot,\cdot\rangle)$, where $\langle\cdot,\cdot\rangle$ represents the canonical inner product of $\mathbb{R}^n$, and let us impose to the deviation vector $\xi(t)$ the initial conditions
\begin{equation}\label{kcc-deviation-xi}
\xi(0)=O\in\mathbb{R}^n, \quad \dot{\xi}(0)=W\neq O,
\end{equation}
where $O$ is the null vector.

For any two vectors $X,Y\in\mathbb{R}^n$ we define an adapted inner product $\langle\langle\cdot,\cdot\rangle\rangle$ to the deviation tensor $\xi$ by
\begin{equation}\label{kcc-deviation-vector}
\langle\langle X,Y\rangle\rangle:=\frac{1}{\langle W,W\rangle}\cdot\langle X,Y\rangle.
\end{equation}
Obviously, $||W||^2:=\langle\langle W,W\rangle\rangle=1$. Thus, the focusing tendency of the trajectories around $t_0=0$ can be described as follows:
\begin{itemize}
  \item trajectories are bunching together if $||\xi(t)||^2<t^2$, $t\approx0^+$,
  \item trajectories are dispersing if $||\xi(t)||^2>t^2$, $t\approx0^+$.
\end{itemize}
The focusing behavior of the trajectories near the origin is shown in Figure \ref{fig:KCC-xi}. See \cite{BChernS2000} for a detailed explanation of the Finslerian geodesics' behavior.

The next result shows that the focusing/dispersing behavior of the trajectories of a second order ODE can also be described from a geometric point of view by using the properties of the deviation curvature tensor. A short proof of it can be found in \cite{SVS22005}.
\begin{lemma}\label{kcc-focusing}
Consider system \eqref{kcc-4} and assume that the deviation vector $\xi(t)$ obeys the set of initial conditions in \eqref{kcc-deviation-xi}. Then, for $t\approx0^+$, the trajectories of \eqref{kcc-4} are
\begin{itemize}
  \item bunching together if and only if the real part of the eigenvalues of $P_j^{i}(0)$ are strictly negative,
  \item dispersing if and only if the real part of the eigenvalues of $P_j^{i}(0)$ are strictly positive.
\end{itemize}
\end{lemma}

\begin{figure}[h]
\begin{minipage}{0.49\linewidth}
\centering
\includegraphics[width=1\linewidth]{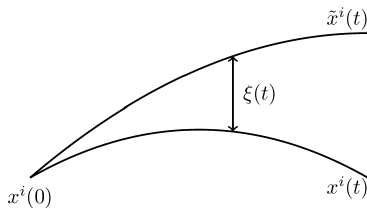}
    \caption*{(a) $||\xi(t)||^2<t^2,t\approx0^+$.}
\end{minipage}
\begin{minipage}{0.49\linewidth}
\centering
\includegraphics[width=1\linewidth]{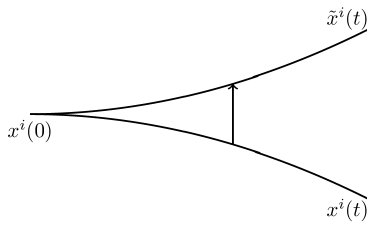}
    \caption*{(b) $||\xi(t)||^2>t^2,t\approx0^+$.}
\end{minipage}
\caption{Behavior of trajectories near zero.}
    \label{fig:KCC-xi}
\end{figure}

Now we can define the concept of Jacobi stability for a second order ODE as follows. This kind of stability refers to the focusing tendency (in a small vicinity of $t_0$) of trajectories of \eqref{kcc-4} with respect to a variation \eqref{kcc-7} that satisfies the condition
\[||x_i(t_0)-\tilde{x}_i(t_0)||=0,\quad ||\dot{x}_i(t_0)-\dot{\tilde{x}}_i(t_0)||\neq0.\]

\begin{definition}\label{kcc-def}
The trajectories of \eqref{kcc-4} are said to be Jacobi stable at $(\boldsymbol{x}(t_0),\dot{\boldsymbol{x}}(t_0))$ if the real parts of the eigenvalues of the deviation curvature tensor $P_j^i|_{t_0}$ are strictly negative, or Jacobi unstable, otherwise.
\end{definition}

We note that Jacobi instability via equation \eqref{kcc-9} means exponential growth of the deviation vector $\xi(t)$, while Jacobi stability means that $\xi(t)$ has an oscillatory variation expressed by a linear combination of trigonometric functions sin and cos. More detailed discussions of the KCC theory, including applications, can be found in \cite{PLA2003,BHS2012,HHLY2015,HPPS2016}.

\section{Main Algorithmic Results}\label{sect-main}

According to the KCC theory described in Section \ref{sect2}, it is necessary to take the following steps to study Jacobi stability of a system of second order ODEs \eqref{eq1-2}.
\vspace{0.1cm}

\noindent{\bf Step 1}. Write a system of second order ODEs in the standard form of the KCC \eqref{eq1-2} (or \eqref{kcc-4}).


\vspace{0.1cm}

\noindent{\bf Step 2}. Compute the deviation curvature tensor $P_j^i$ using \eqref{kcc-10} and derive the symbolic expression of the deviation equations by \eqref{kcc-8}.

\vspace{0.1cm}

\noindent{\bf Step 3}. Determine the sign of the real parts of the eigenvalues of $(P_j^i)$ at each fixed point of the system and study the dynamic behavior of the deviation equations near the fixed points.

\vspace{0.1cm}

Usually, given a system of second order ODEs, {\bf Step 1} can be done by solving a linear system of equations with respect to the variables $(d^2x_1/dt^2,\ldots,d^2x_n/dt^2)$. {\bf Step 2} is mainly based on the formulas \eqref{kcc-8} and \eqref{kcc-10}. In {\bf Step 3}, we will provide a computational approach for detecting Jacobi stability of a system via $(P_j^i)$ at a fixed point. We remark that the study on the dynamic behavior of the deviation equations is performed by numerical simulations (combined with our computational approach and the exact expression of the deviation equations derived in {\bf Step 2}). Next, we divide this section into two parts to describe our main results.

\subsection{Algorithms for Step 2}\label{sect3-1}

This subsection is devoted to providing algorithms to compute the deviation curvature tensor $P_j^i$ and exact expression of the deviation equations. In the following we will present algorithms to implement {\bf Step 2}. We use ``Maple-like'' pseudo-code, based on our Maple implementation.

The algorithm {\bf KCC-invariant}, presented below, is based on \eqref{kcc-10} that can be used to derive the deviation curvature tensor $P_j^i$ for a given system of second order ODEs.

\begin{small}
\begin{algorithm}[H]
\caption{{\bf KCC-invariant}$(G)$}
\hspace*{0.02in} {\bf Input:}
a set of equations $G^i$ in a system of second order ODEs \eqref{eq1-2}\\
\hspace*{0.02in} {\bf Output:}
the deviation curvature tensor $(P_j^i)$
\begin{algorithmic}[1]
\State $n:=\mbox{nop}(G)$;

\For{$i$ {\bf from} 1 {\bf to} $n$}
\For{$j$ {\bf from} 1 {\bf to} $n$}
\For{$\ell$ {\bf from} 1 {\bf to} $n$}
\State $N_j^i:={\partial G^i}/{\partial y_j}$;
\State $G_{j,\ell}^i:={\partial N_j^i}/{\partial y_{\ell}}$;

\EndFor
\EndFor
\EndFor

\For{$i$ {\bf from} 1 {\bf to} $n$}
\For{$j$ {\bf from} 1 {\bf to} $n$}

\State $P_j^i:=-2\cdot\partial G^i/\partial x_j$;

\For{$\ell$ {\bf from} 1 {\bf to} $n$}
\State $P_j^i:=P_j^i+y_{\ell}\cdot\frac{\partial N_j^i}{x_{\ell}}+N_{\ell}^i\cdot N_j^{\ell}-2\cdot G^{\ell}\cdot G_{j,\ell}^i$;

\EndFor
\EndFor
\EndFor

\State \Return $(P_j^i)$;
\end{algorithmic}
\end{algorithm}
\end{small}


The following theorem gives the symbolic expression of the deviation equations \eqref{kcc-8} near a fixed point $\bar{\boldsymbol{x}}$. It also provides a way to deduce the exact solution of the obtained deviation equations.

\begin{theorem}\label{kcc-deviation}
Let $\boldsymbol{\xi}=(\xi_1,\dots,\xi_n)^\mathrm{T}$ be the deviation vector. Then the deviation equations associated to system \eqref{eq1-2} near a fixed point $\bar{\boldsymbol{x}}$ can be written as 
\begin{equation}\label{KCC-xi-2}
    \ddot{\boldsymbol{\xi}}(t)=
    \begin{pmatrix}
       \boldsymbol{A}_{21} & \boldsymbol{A}_{22} 
    \end{pmatrix}
    \begin{pmatrix}
    \boldsymbol{\xi}(t)\\
    \dot{\boldsymbol{\xi}}(t)
    \end{pmatrix},
\end{equation}
where
\begin{equation*}
\begin{split}
\boldsymbol{A}_{21}&=
    \begin{pmatrix}
     -2\frac{\partial G^1}{\partial x_1}(\mu,\bar{\boldsymbol{x}}) & \dots & -2\frac{\partial G^1}{\partial x_n}(\mu,\bar{\boldsymbol{x}})  \\
     \vdots & \ddots & \vdots   \\
     -2\frac{\partial G^n}{\partial x_1}(\mu,\bar{\boldsymbol{x}}) & \dots & -2\frac{\partial G^n}{\partial x_n}(\mu,\bar{\boldsymbol{x}})  \\
    \end{pmatrix},\\
    \boldsymbol{A}_{22}&=
    \begin{pmatrix}
-2N_1^1(\mu,\bar{\boldsymbol{x}}) & \dots & -2N_n^1(\mu,\bar{\boldsymbol{x}}) \\
     \vdots & \ddots & \vdots   \\
     -2N_1^n(\mu,\bar{\boldsymbol{x}}) & \dots & -2N_n^n(\mu,\bar{\boldsymbol{x}}) \\
    \end{pmatrix}.
\end{split}
\end{equation*}
Moreover, the exact solution of $\boldsymbol{\xi}(t)$ under the initial conditions $\boldsymbol{\xi}(0)=(0,\dots,0)^\mathrm{T}, \dot{\boldsymbol{\xi}}(0)=(\xi_{10},\dots,\xi_{n0})^\mathrm{T}$ can be derived by the formula 
\begin{equation}\label{kcc-xi-3}
    \begin{pmatrix}
        \boldsymbol{\xi}(t)\\
        \dot{\boldsymbol{\xi}}(t)
    \end{pmatrix}=
\left(\sum^\infty_{k=0}\frac{(\boldsymbol{A}t)^k}{k!}\right)
\begin{pmatrix}
        \boldsymbol{\xi}(0)\\
        \dot{\boldsymbol{\xi}}(0)
    \end{pmatrix}=\mathrm{exp}(\boldsymbol{A}t)
    \begin{pmatrix}
        \boldsymbol{\xi}(0)\\
        \dot{\boldsymbol{\xi}}(0)
    \end{pmatrix},
\end{equation}
where \begin{equation*}
    \boldsymbol{A}=
    \begin{pmatrix}
     \boldsymbol{0} & \boldsymbol{E} \\
    \boldsymbol{A}_{21} & \boldsymbol{A}_{22}
    \end{pmatrix}.
\end{equation*}
\end{theorem}

\begin{proof}
Note that the equations of deviation direction \eqref{kcc-8} near the fixed point $\bar{\boldsymbol{x}}$ can be rewritten in the following matrix form
\begin{equation}\label{KCC-xi-4}
    \ddot{\boldsymbol{\xi}}(t)=
    \begin{pmatrix}
       \boldsymbol{A}_{21} & \boldsymbol{A}_{22} 
    \end{pmatrix}
    \begin{pmatrix}
    \boldsymbol{\xi}(t)\\
    \dot{\boldsymbol{\xi}}(t)
    \end{pmatrix},
\end{equation}
where
\begin{equation}\label{KCC-xi-5}
    \boldsymbol{A}_{21}=
    \left.\begin{pmatrix}
     -2\frac{\partial G^1}{\partial x_1} & \dots & -2\frac{\partial G^1}{\partial x_n}  \\
     \vdots & \ddots & \vdots   \\
     -2\frac{\partial G^n}{\partial x_1} & \dots & -2\frac{\partial G^n}{\partial x_n}  \\
    \end{pmatrix}\right|_{\bar{x}},
    \quad
    \boldsymbol{A}_{22}=
    \left.\begin{pmatrix}
     -2N_1^1 & \dots & -2N_n^1 \\
     \vdots & \ddots & \vdots   \\
     -2N_1^n & \dots & -2N_n^n \\
    \end{pmatrix}\right|_{\bar{x}}.\nonumber
\end{equation}
By taking the notation
\begin{equation*}
  \begin{split}
        \dot{\boldsymbol{\xi}}(t) = \boldsymbol{\eta}(t), \quad \ddot{\boldsymbol{\xi}}(t) =  \dot{\boldsymbol{\eta}}(t),  
\end{split}
\end{equation*}
we see that system \eqref{KCC-xi-4} is equivalent to the following system of first order equations
\begin{equation}\label{KCC-xi-6}
    \begin{pmatrix}
    \dot{\boldsymbol{\xi}}(t)\\
    \dot{\boldsymbol{\eta}}(t)
    \end{pmatrix}=\boldsymbol{A}
    \begin{pmatrix}
    \boldsymbol{\xi}(t)\\
    \boldsymbol{\eta}(t)
    \end{pmatrix},
\end{equation}
where
\begin{equation}\label{KCC-xi-7}
    \boldsymbol{A}=
    \begin{pmatrix}
     \boldsymbol{0} & \boldsymbol{E} \\
     \boldsymbol{A}_{21} & \boldsymbol{A}_{22}
    \end{pmatrix}.
\end{equation}
Here $\boldsymbol{0}$ is $n\times n$ zero matrix, and $\boldsymbol{E}$ is $n\times n$ identity matrix. Note that \eqref{KCC-xi-6} is a first order differential system with constant coefficient, so that the solution of \eqref{KCC-xi-4} under initial conditions $\boldsymbol{\xi}(0)=(0,\dots,0)^\mathrm{T},\boldsymbol{\eta}(0)=\dot{\boldsymbol{\xi}}(0)=(\xi_{10},\dots,\xi_{n0})^\mathrm{T}$ can be written as the formula \eqref{kcc-xi-3}.
\end{proof}

The following algorithm {\bf KCC-deviation} is based on Theorem \ref{kcc-deviation}, which provides a straightforward calculation method to derive the exact expression of the deviation equations for a considered system of second order ODEs.

\begin{small}
\begin{algorithm}[H]
\caption{{\bf KCC-deviation}$(G)$}
\hspace*{0.02in} {\bf Input:}
a set of equations $G^i$ in a system of second order ODEs \eqref{eq1-2}\\
\hspace*{0.02in} {\bf Output:}
a list of the deviation equations $\xi_i$ in \eqref{kcc-8}
\begin{algorithmic}[1]
\State $n:=\mbox{nop}(G)$;

\For{$i$ {\bf from} 1 {\bf to} $n$}
\For{$j$ {\bf from} 1 {\bf to} $n$}

\State $N_j^i:={\partial G^i}/{\partial y_j}$;

\EndFor
\EndFor

\State $\boldsymbol{\xi}(t)=(\xi_1(t),\ldots,\xi_n(t))^T$;
\State $\dot{\boldsymbol{\xi}}(t)=\left(\mbox{Diff}(\xi_1,t),\ldots,\mbox{Diff}(\xi_n,t)\right)^T$;
\State $\ddot{\boldsymbol{\xi}}(t):=\left(\mbox{Diff}(\xi_1,t\$2),\ldots,\mbox{Diff}(\xi_n,t\$2)\right)^T$;

\State $\mbox{Jac}:=\mbox{Jacobian}(G, [x_1,\ldots,x_n])$;

\State $\Xi(t):=\ddot{\boldsymbol{\xi}}(t)+2\cdot \mbox{Multiply}\left(\big(N_j^i\big),\dot{\boldsymbol{\xi}}(t)\right)+2\cdot \mbox{Multiply}\left(\mbox{Jac},\boldsymbol{\xi}(t)\right)$;

\State \Return $\Xi(t)$;
\end{algorithmic}
\end{algorithm}
\end{small}




\begin{remark}
The termination of the two algorithms is obvious. The correctness of algorithm 1 follows from the formula $P_j^i$ in \eqref{kcc-10}, while the correctness of algorithm 2 follows directly from Theorem \ref{kcc-deviation}. We implemented the above two algorithms in Maple. The Maple program can be found at{\small{
{\textcolor{blue}{\underline{\url{https://github.com/Bo-Math/KCC-Jacobi}}}}}}. Both algorithms are based on standard symbolic tools and hence will run efficiently. In Section \ref{sect4}, we will apply our Maple program to analyze Jacobi stability for several concrete systems of second order ODEs in order to show its feasibility.

\end{remark}

\subsection{Algorithmic Derivation for Step 3}\label{sect3-2}

Our purpose is to derive necessary and sufficient conditions on the parameters $\mu$ for systems of the form \eqref{eq1-2} to have given numbers of Jacobi stable fixed points. This can be done by using the stability criterion of Routh--Hurwitz \cite{RMAM1982,PLMT1985}. Let
\begin{equation}\label{kcc-main-JJ}
\begin{split}
p(\lambda)=\lambda^n+a_1(\mu,\bar{x})\lambda^{n-1}+\cdots+a_n(\mu,\bar{x})
\end{split}
\end{equation}
be the characteristic polynomial of the deviation curvature tensor $(P_j^i)$ at a point $\bar{x}=(x,y)$. The Routh--Hurwitz criterion reduces the problem of determining the negative signs of the real parts of the eigenvalues of $(P_j^i)$ to the problem of determining the signs of certain coefficients $a_i$ of $p(\lambda)$ and the signs of certain determinants $\Delta_j$ of matrices with $a_i$ or 0 as entries.

The necessary and sufficient conditions for $p(\lambda)$ to have all solutions such that ${\rm{Re}}(\lambda)<0$ can be written as
\begin{equation}\label{kcc-main-12}
\begin{split}
&a_n>0,\quad \Delta_1=a_1>0,\quad \Delta_2=\left|
 \begin{matrix}
    a_1& a_3 \\
   1 & a_2 \\
  \end{matrix}
  \right|>0,\\
 &\Delta_3=\left|
 \begin{matrix}
    a_1& a_3 & a_5\\
   1 & a_2 & a_4 \\
   0 & a_1 & a_3 \\
  \end{matrix}
  \right|>0,\ldots,\\
  &\Delta_j=\left|
 \begin{matrix}
    a_1& a_3 & \cdots & \cdots\\
   1 & a_2 & a_4 & \cdots\\
   0 & a_1 & a_3 & \cdots\\
   0 & 1 &a_2 & \cdots\\
   \cdots & \cdots & \cdots & \cdots\\
   0 & 0 & \cdots & a_j\\
  \end{matrix}
  \right|>0,\ldots,\\
\end{split}
\end{equation}
for all $j=1,2,\ldots,n$. Here $\Delta_1,\ldots,\Delta_n$ are known as the \textit{Hurwitz determinants} of $p(\lambda)$.

The following result gives necessary and sufficient conditions for a given system of the form \eqref{eq1-2} to have a prescribed number of Jacobi stable fixed points.

\begin{theorem}\label{semi-kcc}
For a $n$-dimensional system of second order ODEs \eqref{eq1-2}, the system has exactly $k$ Jacobi stable fixed points if and only if the following semi-algebraic system
\begin{equation}\label{kcc-main-13}
\left\{\begin{array}{l}
G^{1,1}(\mu,x,0)=G^{2,1}(\mu,x,0)
=\cdots=G^{n,1}(\mu,x,0)=0,\\ 
G^{i,2}(\mu,x,0)\neq0,\quad i=1,\ldots,n,\\
a_{n,1}(\mu,x,0)\cdot a_{n,2}(\mu,x,0)>0,\\
\Delta_{j,1}(\mu,x,0)\cdot\Delta_{j,2}(\mu,x,0)>0,\quad j=1,\ldots,n\\
\end{array}\right.
\end{equation}
has exactly $k$ distinct real solutions with respect to the variable $x=(x_1,\ldots,x_n)$. Here $G^{i,1}$ and $G^{i,2}$, $a_{n,1}$ and $a_{n,2}$, $\Delta_{j,1}$ and $\Delta_{j,2}$ are respectively the numerator and denominator of the function $G^i$, $a_{n}$, and $\Delta_{j}$, with $\mu=(\mu_1,\ldots,\mu_p)$ are parameters appearing in \eqref{eq1-2}.
\end{theorem}

\begin{proof}
Let $\bar{x}=(x_i,dx_i/dt)$ be the fixed point of \eqref{eq1-2}. Note that at the point we have $\dot{x}_1=\cdots=\dot{x}_n=0$. The desired result follows directly from Routh--Hurwitz criterion.
\end{proof}



The main task of Theorem \ref{semi-kcc} is to find conditions on the parameters $\mu$ for system \eqref{kcc-main-13} to have exactly $k$ distinct real solutions. Now we describe the main steps for automatically analyzing Theorem \ref{semi-kcc} by using methods from symbolic computation.
\vspace{0.1cm}

{\bf Step 1}. Formulate the semi-algebraic system \eqref{kcc-main-13} from a system of second order ODEs \eqref{eq1-2} by using Algorithm 1 introduced in Section \ref{sect3-1}. Denote by $\mathcal{S}$ the semi-algebraic system for solving, $\Psi$ the set of inequalities of $\mathcal{S}$, $\mathcal{F}$ the set of polynomials in $\Psi$, and $\mathcal{P}$ the set of polynomials in the equations of $\mathcal{S}$.
\vspace{0.1cm}

{\bf Step 2}. Triangularize the set $\mathcal{P}$ of polynomials to obtain one or several (regular) triangular sets $\mathcal{T}_{k}$ by using the method of triangular decomposition or
Gr\"obner bases.

\vspace{0.1cm}

{\bf Step 3}. For each triangular set $\mathcal{T}_{k}$, use the polynomial set $\mathcal{F}$ to compute an algebraic variety $V$ in $\mu$ by means of real solution classification (e.g., Yang--Xia's method \cite{YHX01,LYBX05} or Lazard--Rouillier's method
\cite{DLFR}), which decomposes the parameter space $\mathbb{R}^{p}$ into finitely many cells such that in each cell the number of real zeros of $\mathcal{T}_{k}$ and the signs of polynomials in $\mathcal{F}$ at these real zeros remain invariant. The algebraic variety is defined by polynomials in $\mu$. Then take a rational sample point from each cell by using the method of PCAD or critical points \cite{msed2007}, and isolate the real zeros of $\mathcal{T}_{k}$ by rational intervals at this sample point. In this way, the number of real zeros of $\mathcal{T}_{k}$ and the signs of polynomials in $\mathcal{F}$ at these real zeros in each cell are determined.
\vspace{0.1cm}

{\bf Step 4}. Determine the signs of (the factors of) the defining polynomials of $V$ at each sample point. Formulate the conditions on $\mu$ according to the signs of these defining polynomials at the sample points in those cells in which the system $S$ has exactly the number of real solutions we want.

\vspace{0.1cm}

{\bf Step 5}. Output the conditions on $\mu$ such that the system has a prescribed number of Jacobi stable fixed points.

\vspace{0.1cm}


There are several software packages that can be used for solving semi-algebraic systems of the form \eqref{kcc-main-13}. For example, the Maple package {\sf DV} by Moroz and Rouillier, and the Maple package {\sf DISCOVERER} (see also recent improvements in the Maple package RegularChains[SemiAlgebraicSe- tTools]), developed by Xia. These packages implement the method of discriminant varieties of Lazard and Rouillier \cite{DLFR} and the method of Yang and Xia \cite{LYBX05} for real solution classification. In Section \ref{sect4}, we will present several examples to demonstrate the applicability and the computational efficiency of our general algorithmic approach. In practice, we will use the output results of Theorem \ref{semi-kcc} to study the focusing/dispersing behavior of the trajectories near the fixed points of a system of the form \eqref{eq1-2}. The dynamic behaviors of the resulting deviation equations near the fixed points will also be performed by numerical simulations.


\section{Examples and Experiments}\label{sect4}
In this section, we explain how to apply the algorithmic tests to study Jacobi stability of systems of second order ODEs and illustrate some of the computational steps by a model of wound strings. In addition, using our algorithmic approach, we also analyze Jacobi stability of an airfoil model with cubic nonlinearity in supersonic flow and a 3-DOF tractor seat-operator model. The experimental results show the applicability and efficiency of our algorithmic approach. All the computations were made in Maple 2023, running under a Windows 11 laptop on Intel(R) Core(TM) Ultra 7 155H  1.40 GHz.

\subsection{Illustrative Example: A Model of Wound Strings}\label{sect5.1}

Lake and Harko proposed a model of wound strings in 2016 \cite{Matthew2016}. It is a system of second order ODEs, expressed as follows
\begin{equation}\label{WS-1}
    \begin{split}
        \ddot{\rho}&+\frac{\rho}{a^2\rho^2+m^2R^2}\left[a^2\left(1-\dot{\rho}^2\right)-\dot{R}^2-C^2\frac{\left(a^2\rho^2+m^2R^2\right)^2}{\rho^4R^2}\right]=0,\\
        \ddot{R}&+\frac{R}{a^2\rho^2+m^2R^2}\left[m^2a^2(1-\dot{\rho}^2)-m^2\dot{R}^2-a^2C^2\frac{(a^2\rho^2+m^2R^2)^2}{\rho^2R^4}\right]=0,
    \end{split}
\end{equation}
where $\rho$ and $R$ are variables,  $a^2\in(0,1]$ is a phenomenological ``wrap factor'', $C$ and $m$ are arbitrary real constants. For brevity, we denote the parameters by $\mu=(a,C,m)$.

To write the model in the standard form of the KCC, we change the notation as
\begin{equation*}
    \rho = x_1, \quad R = x_2, \quad \dot{\rho} = y_1, \quad \dot{R} = y_2.
\end{equation*}
Let $(x,y)=(x_1,x_2,y_1,y_2)$. Then system \eqref{WS-1} becomes
\begin{equation}\label{WS-2}
\begin{split}
\ddot{x}_1+2G^1(\mu,x,y)=0,\quad
\ddot{x}_2+2G^2(\mu,x,y)=0,
\end{split}
\end{equation}
where
\begin{equation}\label{WS-3}
\begin{split}
&G^1=\frac{x_1}{2a^2x_1^2+2m^2x_2^2}\Bigg[a^2\left(1-y_1^2\right)-y_2^2-C^2\frac{\left(a^2x_1^2+m^2x_2^2\right)^2}{x_1^4x_2^2}\Bigg],\\
&G^2=\frac{x_2}{2a^2x_1^2+2m^2x_2^2}\Bigg[
m^2\Big(a^2(1-y_1^2)-y_2^2\Big)-a^2C^2\frac{(a^2x_1^2+m^2x_2^2)^2}{x_1^2x_2^4}\Bigg].
\end{split}
\end{equation}

Applying the algorithm {\bf KCC-invariant} to \eqref{WS-3}, we obtain the following deviation curvature tensor $(P^i_j)$:
\begin{equation}\label{WS-4}
\begin{split}
P^1_1&=-\frac{1}{x_1^4x_2^2\left(a^2x_1^2+m^2x_2^2\right)^2}\Big(2 C^{2} a^{6} x_{1}^{6}+7 C^{2} a^{4} m^{2} x_{1}^{4} x_{2}^{2}\\
&+8 C^{2} a^{2} m^{4} x_{1}^{2} x_{2}^{4}+3 C^{2} m^{6} x_{2}^{6}-3 y_{2} x_{1}^{5} a^{2} y_{1} m^{2} x_{2}^{3}-2 a^{4} x_{1}^{6} x_{2}^{2}\\
&+a^{2} m^{2} x_{1}^{4} x_{2}^{4}+2 a^{2} x_{1}^{6} x_{2}^{2} y_{2}^{2}-m^{2} x_{1}^{4} x_{2}^{4} y_{2}^{2}\Big),\\
P^1_2&=-\frac{1}{x_1x_2^3\left(a^2x_1^2+m^2x_2^2\right)^2}\Big(3 C^{2} a^{6} x_{1}^{4}+6 C^{2} a^{4} m^{2} x_{1}^{2} x_{2}^{2}\\
&+3 C^{2} a^{2} m^{4} x_{2}^{4}+3 a^{2} m^{2} x_{1}^{2} x_{2}^{4} y_{1}^{2}-3 a^{2} m^{2} x_{1}^{2} x_{2}^{4}\\
&-2 a^{2} x_{1}^{3} x_{2}^{3} y_{1} y_{2}+m^{2} x_{1} x_{2}^{5} y_{1} y_{2}\Big),\\
P^2_1&=-\frac{a^2m^2}{x_1^3x_2\left(a^2x_1^2+m^2x_2^2\right)^2}\Big(3 C^{2} a^{4} x_{1}^{4}+6 C^{2} a^{2} m^{2} x_{1}^{2} x_{2}^{2}+3 C^{2} m^{4} x_{2}^{4}\\
&+a^{2} x_{1}^{5} x_{2} y_{1} y_{2}-2 y_{2} m^{2} y_{1} x_{1}^{3} x_{2}^{3}-3 a^{2} x_{1}^{4} x_{2}^{2}+3 y_{2}^{2} x_{1}^{4} x_{2}^{2}\Big),\\
P^2_2&=-\frac{a^2}{x_1^2x_2^4\left(a^2x_1^2+m^2x_2^2\right)^2}\Big(3 C^{2} a^{6} x_{1}^{6}+8 C^{2} a^{4} m^{2} x_{1}^{4} x_{2}^{2}\\
&+7 C^{2} a^{2} m^{4} x_{1}^{2} x_{2}^{4}+2 C^{2} m^{6} x_{2}^{6}-y_{1}^{2} x_{2}^{4} m^{2} a^{2} x_{1}^{4}+2 m^{4} x_{1}^{2} x_{2}^{6} y_{1}^{2}\\
&+a^{2} m^{2} x_{1}^{4} x_{2}^{4}-2 m^{4} x_{1}^{2} x_{2}^{6}-3 m^{2} x_{1}^{3} x_{2}^{5} y_{1} y_{2}
\Big).
\end{split}
\end{equation}

The obtained expressions in \eqref{WS-4} are different from those in \cite[Section 5]{Matthew2016}. After comparing the calculation process, we found that $G^2$ in \cite{Matthew2016} was incorrect ($x^2$ should be $X^2$ there). As a consequence, the symbol $X^2$ was mistakenly written as $X^1$ in $N^2_1$ and $N^2_2$. This ultimately led to errors in the $(P^i_j)$. Thus, some calculations of the deviation curvature tensor in \cite{Matthew2016} need to be reconsidered algorithmically, using the algorithm developed in this
paper.

Now consider system \eqref{WS-2} and let $\bar{x}=(x_1,x_2,y_1=0,y_2=0)$ be a fixed point of it. We have
\begin{equation}
\begin{split}
G^1(\mu,\bar{x})&=-\frac{1}{2 x_{1}^{3} \left(a^{2} x_{1}^{2}+m^{2} x_{2}^{2}\right) x_{2}^{2}}G^{1,1}(\mu,\bar{x}),  \\
G^2(\mu,\bar{x})&=-\frac{a^2}{2 x_{2}^{3} \left(a^{2} x_{1}^{2}+m^{2} x_{2}^{2}\right) x_{1}^{2}}G^{2,1}(\mu,\bar{x}), 
\end{split}
\end{equation}
where
\begin{equation}
\begin{split}
G^{1,1}(\mu,\bar{x})&=C^{2} a^{4} x_{1}^{4}+2 C^{2} a^{2} m^{2} x_{1}^{2} x_{2}^{2}+C^{2} m^{4} x_{2}^{4}-a^{2} x_{1}^{4} x_{2}^{2},\\
G^{2,1}(\mu,\bar{x})&=C^{2} a^{4} x_{1}^{4}+2 C^{2} a^{2} m^{2} x_{1}^{2} x_{2}^{2}+C^{2} m^{4} x_{2}^{4}-m^{2} x_{1}^{2} x_{2}^{4}.\nonumber
\end{split}
\end{equation}
The characteristic polynomial of $(P_j^i)$ at the fixed point is
\begin{equation}\label{kcc-main-JJ}
\begin{split}
p(\lambda)=\lambda^2+a_1(\mu,\bar{x})\lambda+a_2(\mu,\bar{x}),
\end{split}
\end{equation}
where
\begin{equation}
\begin{split}
a_1(\mu,\bar{x})&=\frac{1}{\left(a^{2} x_{1}^{2}+m^{2} x_{2}^{2}\right)^{2} x_{1}^{4} x_{2}^{4}}a_{1,1}(\mu,\bar{x}),\\
a_2(\mu,\bar{x})&=\frac{2a^2}{\left(a^{2} x_{1}^{2}+m^{2} x_{2}^{2}\right)^{2} x_{1}^{6} x_{2}^{6}}a_{2,1}(\mu,\bar{x}),\nonumber
\end{split}
\end{equation}
with
\begin{equation}
\begin{split}
a_{1,1}&(\mu,\bar{x})=3 C^{2} a^{8} x_{1}^{8}+8 C^{2} a^{6} m^{2} x_{1}^{6} x_{2}^{2}+7 C^{2} a^{4} m^{4} x_{1}^{4} x_{2}^{4}\\
&+2 C^{2} a^{2} m^{6} x_{1}^{2} x_{2}^{6}+2 C^{2} a^{6} x_{1}^{6} x_{2}^{2}+7 C^{2} a^{4} m^{2} x_{1}^{4} x_{2}^{4}+8 C^{2} a^{2} m^{4} x_{1}^{2} x_{2}^{6}\\
&+3 C^{2} m^{6} x_{2}^{8}+a^{4} m^{2} x_{1}^{6} x_{2}^{4}-2 a^{2} m^{4} x_{1}^{4} x_{2}^{6}-2 a^{4} x_{1}^{6} x_{2}^{4}+a^{2} m^{2} x_{1}^{4} x_{2}^{6},\\
a_{2,1}&(\mu,\bar{x})=3 C^{4} a^{8} x_{1}^{8}+8 C^{4} a^{6} m^{2} x_{1}^{6} x_{2}^{2}+10 C^{4} a^{4} m^{4} x_{1}^{4} x_{2}^{4}\\
&+8 C^{4} a^{2} m^{6} x_{1}^{2} x_{2}^{6}+3 C^{4} m^{8} x_{2}^{8}-3 C^{2} a^{6} x_{1}^{8} x_{2}^{2}\\
&+5 C^{2} a^{4} m^{2} x_{1}^{6} x_{2}^{4}+5 C^{2} a^{2} m^{4} x_{1}^{4} x_{2}^{6}-3 C^{2} m^{6} x_{1}^{2} x_{2}^{8}-a^{2} m^{2} x_{1}^{6} x_{2}^{6}.\nonumber
\end{split}
\end{equation}

According to Theorem \ref{semi-kcc}, we formulate the following algebraic system $\mathcal{S}$ for the study of Jacobi fixed points of system \eqref{WS-2}:
\begin{equation}
\mathcal{S}:=
\left\{\begin{array}{l}
G^{1,1}(\mu,\bar{x})=G^{2,1}(\mu,\bar{x})
=0,\\ 
x_1\neq0,\,\, x_2\neq0,\,\, a^2>0,\,\, 1-a^2\geq0,\\
a_{1,1}(\mu,\bar{x})>0,\,\, a_{2,1}(\mu,\bar{x})>0.\\
\end{array}\right.\nonumber
\end{equation}
Using the command RealRootClassification in the Maple package RegularChains[SemiAlgebraicSe- tTools] to the above system, we see that system $\mathcal{S}$ cannot have one, two, or three real solutions; system $\mathcal{S}$ always has four real solutions when $m\neq0$, $C\neq0$, and $a^2-1\neq0$.

Here we restate the result related to the Jacobi stability of the model \eqref{WS-1} as follows.

\begin{theorem}\label{main-kcc1}
The model \eqref{WS-1} cannot have one, two, or three Jacobi stable fixed points, and always has four Jacobi stable fixed points when the condition
\[\mathcal{C}_0=[m\neq0, \quad C\neq0,\quad  a^2-1\neq0]\]
holds.
\end{theorem}

\begin{remark}
Theorem \ref{main-kcc1} is quite different from the conclusion made in \cite[Theorem 1]{Matthew2016}. To verify the correctness of our result, we will perform numerical simulations of dynamic behaviors of the model \eqref{WS-1} near the fixed points. We take the parameter values $(a,C,m)=(\frac{1}{2},1,-1)$. According to Theorem \ref{main-kcc1}, the model \eqref{WS-2} has four Jacobi stable fixed points $(\bar{x}_1,\bar{x}_2)=(\pm 2,\pm 1)$ and $(\bar{x}_1,\bar{x}_2)=(\pm 2,\mp 1)$. However, by \cite[Theorem 1]{Matthew2016}, we see that the fixed points $(-2,-1)$ and $(2,1)$ are Jacobi stable, while the other two fixed points $(-2,1)$ and $(2,-1)$ are Jacobi unstable. In the following we will show that both of the two fixed points $(-2,1)$ and $(2,-1)$ are also Jacobi stable.

\end{remark}

Taking $(a,C,m)=(\frac{1}{2},1,-1)$ and applying algorithm {\bf KCC-deviation} to the $G^i$ in \eqref{WS-2}, we obtain the following deviation equations:
\begin{equation}\label{WS-dev-1}
    \begin{split}
        &\frac{d^2\xi_1}{dt^2}-\frac{2x_1y_1}{x_1^2+4x_2^2}\frac{d\xi_1}{dt}-\frac{8x_1y_2}{x_1^2+4x_2^2}\frac{d\xi_2}{dt}+M^1_1\xi_1+M^1_2\xi_2=0,\\
        &\frac{d^2\xi_2}{dt^2}-\frac{2x_2y_1}{x_1^2+4x_2^2}\frac{d\xi_1}{dt}-\frac{8x_2y_2}{x_1^2+4x_2^2}\frac{d\xi_2}{dt}+M^2_1\xi_1+M^2_2\xi_2=0,
    \end{split}
\end{equation}
where
\begin{equation*}
\begin{split}
M^1_1&=\frac{1}{4 \left(x_{1}^{2}+4 x_{2}^{2}\right)^{2} x_{1}^{4} x_{2}^{2}}\Big(4 x_{1}^{6} x_{2}^{2} y_{1}^{2}+16 x_{1}^{6} x_{2}^{2} y_{2}^{2}-16 x_{1}^{4} x_{2}^{4} y_{1}^{2}\\
&-64 x_{1}^{4} x_{2}^{4} y_{2}^{2}-4 x_{1}^{6} x_{2}^{2}+16 x_{1}^{4} x_{2}^{4}+x_{1}^{6}+20 x_{1}^{4} x_{2}^{2}+112 x_{1}^{2} x_{2}^{4}+192 x_{2}^{6}\Big),\\
M^1_2&=\frac{1}{2 x_{1} \left(x_{1}^{2}+4 x_{2}^{2}\right)^{2} x_{2}^{3}}\Big(16 x_{1}^{2} x_{2}^{4} y_{1}^{2}+64 x_{1}^{2} x_{2}^{4} y_{2}^{2}-16 x_{1}^{2} x_{2}^{4}\\
&+x_{1}^{4}+8 x_{1}^{2} x_{2}^{2}+16 x_{2}^{4}\Big),\\
M^2_1&=\frac{1}{2 x_{2} \left(x_{1}^{2}+4 x_{2}^{2}\right)^{2} x_{1}^{3}}\Big(4 x_{1}^{4} x_{2}^{2} y_{1}^{2}+16 x_{1}^{4} x_{2}^{2} y_{2}^{2}-4 x_{1}^{4} x_{2}^{2}+x_{1}^{4}\\
&+8 x_{1}^{2} x_{2}^{2}+16 x_{2}^{4}\Big),\\
M^2_2&=\frac{1}{16 \left(x_{1}^{2}+4 x_{2}^{2}\right)^{2} x_{1}^{2} x_{2}^{4}}\Big(-16 x_{1}^{4} x_{2}^{4} y_{1}^{2}-64 x_{1}^{4} x_{2}^{4} y_{2}^{2}+64 x_{1}^{2} x_{2}^{6} y_{1}^{2}\\
&+256 x_{1}^{2} x_{2}^{6} y_{2}^{2}+16 x_{1}^{4} x_{2}^{4}-64 x_{1}^{2} x_{2}^{6}+3 x_{1}^{6}+28 x_{1}^{4} x_{2}^{2}+80 x_{1}^{2} x_{2}^{4}+64 x_{2}^{6}\Big).
\end{split}
\end{equation*}

By evaluating system \eqref{WS-dev-1} at the fixed points $(x_1,x_2,y_1,y_2)=(-2,1,0,0)$ and $(x_1,x_2,y_1,y_2)=(2,-1,0,0)$ respectively, we find that the corresponding deviation equations have the same expression:
\begin{equation}\label{WS-dev-2}
\begin{split}
\frac{d^2\xi_1}{dt^2}+\frac{1}{4}\xi_1=0,\quad
\frac{d^2\xi_2}{dt^2}+\frac{1}{4}\xi_2=0.
\end{split}
\end{equation}

A direct computation shows that the general solutions to these equations are given by
\begin{equation}\label{WS-dev-3}
\begin{split}
\xi_1(t)=2\xi_{10}\sin\Big(\frac{t}{2}\Big),\quad
\xi_2(t)=2\xi_{20}\sin\Big(\frac{t}{2}\Big),
\end{split}
\end{equation}
where we have used the initial conditions $\xi_1(0)=\xi_2(0)=0$, $\dot{\xi}_1(0)=\xi_{10}$, $\dot{\xi}_2(0)=\xi_{20}$. It follows from \eqref{kcc-deviation-vector} that
\begin{equation}\label{WS-dev-4}
\begin{split}
||\xi(t)||^2=\frac{1}{\xi_{10}^2+\xi_{20}^2}\cdot\Big[(\xi_1(t))^2+(\xi_2(t))^2\Big]=4\sin^2\left(\frac{t}{2}\right)<t^2,\,\, t\approx0^+.\nonumber
\end{split}
\end{equation}
According to the KCC theory in Section \ref{sect2}, we prove that the fixed points $(-2,1)$ and $(2,-1)$ are also Jacobi stable.

In fact, it is easy to check that 
the expressions of the deviation equations in \eqref{WS-dev-1} are the same at the four fixed points $(\pm 2,\pm 1)$ and $(\pm 2,\mp 1)$. That is to say, these four fixed points have the same stability. This observation further verifies our conclusion. In the following, based on our computational results, we study the focusing/dispersing
behavior of the trajectories near the fixed points of the model \eqref{WS-1}. In order to save space, the dynamic behaviors of the model \eqref{WS-1} near the fixed points are put in Appendix \ref{sec-A}.

\subsection{Other Models and Remarks}\label{sect4.2}

\subsubsection{An Airfoil Model with Cubic Nonlinearity in Supersonic Flow}\label{sect4.2-1}

The aeroelastic model with cubic nonlinear pitching stiffness can be expressed as 
\begin{equation}\label{AM-1}
    \begin{split}
        &m\ddot{h}+S_{\alpha}\ddot{\alpha}+c_h\dot{h}+K_h h =L, \\
        &S_{\alpha}\ddot{h}+ I_{\alpha}\ddot{\alpha}+c_{\alpha}\dot{\alpha}+K_{\alpha}\alpha+eK_{\alpha}\alpha^3 = M,
    \end{split}
\end{equation}
where $m$ is the mass of the wing, $S_{\alpha}=mx_{\alpha}b$ is the wing mass static moment about the elastic axis, $I_{\alpha}=m(r_{\alpha}b)^2$ is the inertial moment of the wing about the elastic axis, $c_{h}$ and $c_{\alpha}$ are the linear plunging and pitching damping coefficients, $K_h$ and $K_{\alpha}$ are the plunging and pitching stiffness coefficients, $e$ is
the non-dimensional nonlinear stiffness coefficient, $b$ is the wing's semi-chord length, $h$ is the plunging displacement, $\alpha$ is the pitching angle, the superposed dots denote derivative with respect to time $t$, $L$ and $M$ are the aerodynamic force and moment. We refer to \cite{Librescu2003,Zheng2008} for the setting details of this example.

Introducing the following dimensionless variable
\begin{equation}
\begin{split}
\eta&=\frac{h}{b},\quad \tau=\frac{U_{\infty}}{b}t,\quad 
\omega_h=\sqrt{\frac{K_h}{m}},\quad \omega_{\alpha}=\sqrt{\frac{K_{\alpha}}{I_{\alpha}}},\\
\bar{\omega}&=\frac{\omega_h}{\omega_{\alpha}},\quad 
\mu=\frac{m}{4\rho_{\infty}b^2},\quad V=\frac{U_{\infty}}{b\omega_{\alpha}}\nonumber
\end{split}
\end{equation}
the dimensionless equations of motion for the airfoil can be written as (see \cite{Zheng2008} for the setting details)
\begin{equation}\label{AM-2}
\begin{split}
\ddot{\eta}+x_{\alpha} \ddot{\alpha}&+ 2\zeta_h \frac{\bar{\omega}}{V}\dot{\eta}+\left(\frac{\bar{\omega}}{V}\right)^2\eta=-\frac{1}{\mu M_{\infty}}\Bigg[\alpha+\dot{\eta}+(1-a)\dot{\alpha}\\&\qquad\quad+\frac{1}{12}M_{\infty}^2(1+k)\alpha^3\Bigg], \\
\frac{x_{\alpha}}{r_{\alpha}^2}\ddot{\eta}+\ddot{\alpha}& +2\zeta_{\alpha}\frac{1}{V}\dot{\alpha}+\frac{1}{V^2}\alpha+\frac{e}{V^2}\alpha^3=\\
        &\frac{1}{\mu M_{\infty}r_{\alpha}^2}\Bigg[(1-a)\alpha+(1-a)\dot{\eta}+\left(\frac{4}{3}-2a+a^2\right)\dot{\alpha}\\
        &+\frac{1}{12}M_{\infty}^2(1-a)(1+k)\alpha^3\Bigg],
    \end{split}
\end{equation}
where $M_{\infty}$ is the flight Mach number, $k$ is the isentropic gas coefficient. For easy reference and comparison, we take some of the parameter values for the airfoil constants suggested in \cite{Zheng2008}
\begin{equation}
\begin{split}
\mu&=50,\quad \bar{\omega}=1, \quad x_{\alpha}=\frac{1}{4},\quad a=r_{\alpha}^2=\frac{1}{2},\\
\zeta_h&=\zeta_{\alpha}=\frac{1}{10},\quad e=20,\quad k=\frac{7}{5}.\nonumber
\end{split}
\end{equation}
Let $\eta=x_1$, $\alpha=x_2$, $\dot{\eta}=y_1$, $\dot{\alpha}=y_2$, set $(x,y)=(x_1,x_2,y_1,y_2)$ and $\bar{\mu}=(M_{\infty},V)$. Then system \eqref{AM-2} can be further transformed into the following KCC standard form:
\begin{equation}\label{AM-3}
\begin{split}
\ddot{x}_1+2G^1(\bar{\mu},x,y)=0,\quad
\ddot{x}_2+2G^2(\bar{\mu},x,y)=0,
\end{split}
\end{equation}
where
\begin{equation}
\begin{split}
G^1&(\bar{\mu},x,y)=\frac{1}{2100 V^{2}M_{\infty}}\Big(6 V^{2} M_{\infty}^{2} x_{2}^{3}-6000 M_{\infty} x_{2}^{3}+30 V^{2} x_{2}+30 V^{2} y_{1}\\
&+19 V^{2} y_{2}+240 V M_{\infty} y_{1}-60 V M_{\infty} y_{2}+1200 M_{\infty} x_{1}-300 M_{\infty} x_{2}\Big),\\
G^2&(\bar{\mu},x,y)=-\frac{1}{5250 V^{2}M_{\infty}}\Big(18 V^{2} M_{\infty}^{2} x_{2}^{3}-60000 M_{\infty} x_{2}^{3}+90 V^{2} x_{2}+90 V^{2} y_{1}\\
&+85 V^{2} y_{2}+300 V M_{\infty} y_{1}-600 V M_{\infty} y_{2}+1500 M_{\infty} x_{1}-3000 M_{\infty} x_{2}\Big).\nonumber
\end{split}
\end{equation}

Our computational results of the model \eqref{AM-3} are summarized as the following theorem.
\begin{theorem}\label{KCC-math-airfoil}
The following statements hold for the aeroelastic airfoil model \eqref{AM-3}.

\begin{itemize}
    \item[(a)] System \eqref{AM-3} has exactly one Jacobi stable fixed point if and only if one of the following three conditions holds
\begin{equation}
\begin{split}  
\mathcal{C}_1=&[0<R_1,\,\, 0<R_2,\,\,0<R_4,\,\, 0<R_5,\,\, 0<R_6] \wedge \mathcal{C}^*,\\
\mathcal{C}_2=&[0<R_1,\,\, 0<R_2,\,\,0<R_3,\,\, R_4<0] \wedge \mathcal{C}^*,\\
\mathcal{C}_3=&[0<R_1,\,\, 0<R_2,\,\,0<R_3,\,\,0<R_4,\,\,R_5<0,\,\,0<R_6]\wedge \mathcal{C}^*,\nonumber
\end{split}
\end{equation}
where
\begin{equation}
\begin{split}
&R_1=-9450 V^{2} M_{\infty}-43 V^{2}+135 V M_{\infty}+621900 M_{\infty}^{2},\\
&R_2=1800 V^{4} M_{\infty}+1500 V^{3} M_{\infty}^{2}-15660000 V^{2} M_{\infty}^{3}+4 V^{4}+20 V^{3} M_{\infty}\\
&-123675 V^{2} M_{\infty}^{2}+297000 V M_{\infty}^{3}+769590000 M_{\infty}^{4},\\
&R_3=-V^{2}+50 M_{\infty},\\
&R_4=V^{2} M_{\infty}-5000,\\
&R_5=-18900 V^{4} M_{\infty}^{2}+43 V^{4} M_{\infty}-135 V^{3} M_{\infty}^{2}+795600 V^{2} M_{\infty}^{3}\\
&+47250000 V^{2} M_{\infty}-215000 V^{2}+675000 V M_{\infty}-1615500000 M_{\infty}^{2},\\
&R_6=3600 V^{6} M_{\infty}^{2}+3000 V^{5} M_{\infty}^{3}-31320000 V^{4} M_{\infty}^{4}-4 V^{6} M_{\infty}-20 V^{5} M_{\infty}^{2}\\
&-146325 V^{4} M_{\infty}^{3}-522000 V^{3} M_{\infty}^{4}+1579410000 V^{2} M_{\infty}^{5}-4500000 V^{4} M_{\infty}\\
&-532500000 V^{3} M_{\infty}^{2}+155250000000 V^{2} M_{\infty}^{3}+20000 V^{4}+100000 V^{3} M_{\infty}\\
&+56625000 V^{2} M_{\infty}^{2}+28485000000 V M_{\infty}^{3}-7829550000000 M_{\infty}^{4},\\
&\mathcal{C}^*=[(M_{\infty}-10)\cdot R_1\cdot R_2\cdot R_3\cdot R_4\cdot R_5\cdot R_6\neq0].
\nonumber
\end{split}
\end{equation}
    
    \item[(b)] System \eqref{AM-3} has exactly two Jacobi stable fixed points if and only if one of the following two conditions holds
\begin{equation}
\begin{split}  
\mathcal{C}_4=&[R_1<0,\,\, R_3<0,\,\,R_4<0,\,\, 0<R_5,\,\, 0<R_6] \wedge \mathcal{C}^*,\\
\mathcal{C}_5=&[0<R_1,\,\, R_2<0,\,\,R_3<0,\,\,R_4<0,\,\, 0<R_5,\,\, 0<R_6] \wedge \mathcal{C}^*.\nonumber
\end{split}
\end{equation}

\end{itemize}

\end{theorem}

In Figure \ref{fig:kcc-equilibria}, we provide partitions of the parameter sets $\mathcal{C}_1$ and $\mathcal{C}_5$ for distinct numbers of Jacobi stable fixed points. The focusing/dispersing
behaviors of the trajectories near the fixed points of the model \eqref{AM-3} are given in Appendix \ref{sec-B}.

\begin{figure}[htbp]
 \centering
\includegraphics[width=0.4\textwidth]{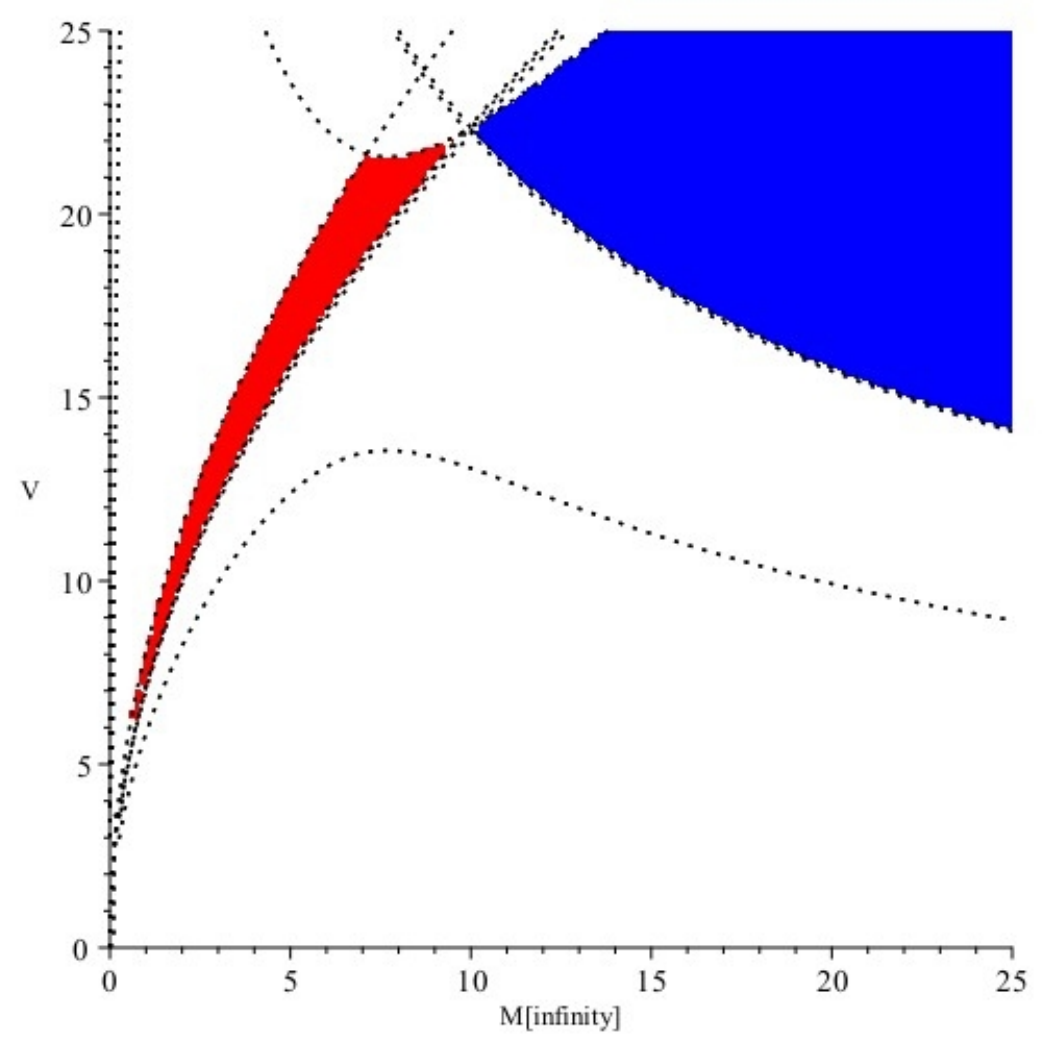}
    \caption{In the blue and red regions, system \eqref{AM-3} has one and two Jacobi stable fixed points, respectively. }\label{fig:kcc-equilibria}
\end{figure}

\subsubsection{A 3-DOF Tractor Seat-Operator Model}\label{sect4.2-2}

In order to save space, the results of the model are given in Appendix \ref{sec-C}.

\section{Conclusion and Future Work}\label{sect5}

We develop a symbolic Maple program for computing the second KCC invariant for systems of second order ODEs in arbitrary dimension. Our program allowed us to detect errors in previous hand-guided computations. With the aid of this program, we reduce the analysis of Jacobi stability to the solution of a semi-algebraic system and introduce a systematical computational approach for rigorously analyzing the conditions on the parameters under which a considered system of second order ODEs has given numbers of Jacobi stable fixed points. We present some results about the Jacobi stability of certain differential models, including the model of wound strings, a two-dimensional airfoil model with cubic nonlinearity in supersonic flow and the 3-DOF tractor seat-operator model. The results of experiments performed demonstrate the effectiveness of the proposed algorithmic approach.

As mentioned in Section \ref{sect-main}, in the {\bf Step 1} of the Jacobi stability analysis. Currently, we focus on the systems of second order ODEs that are linear with respect to the variables $(d^2x_1/dt^2,\ldots,d^2x_n/dt^2)$. For nonlinear systems of second order ODEs with respect to these variables, we need to combine advanced techniques from symbolic computation, such as Gr{\"o}bner basis or triangular decomposition, to solve such nonlinear equations. Exploring differential elimination methods to obtain KCC standard forms from a system of second order ODEs is also one of the future research directions.

It would be interesting to employ our algorithmic approach for analyzing Jacobi stability for systems of second order ODEs from different domains of science and engineering. The involved symbolic computation may easily become too heavy and intractable as the number and degree of equations and the number of system parameters increase. It is therefore necessary to improve the methods and software package we use, to combine them with other methods and tools, and to develop specialized techniques that take into account the structure of the considered systems. Our research along these directions is in progress.

%

\section*{Acknowledgments}
B. Huang's work is supported by the National Natural Science Foundation of China (NSFC 12101032 and NSFC 12131004), and the Fundamental Research Funds for the Central Universities. D. Wang's work is supported by Beijing Advanced Innovation Center for
Future Blockchain and Privacy Computing. X. Wang wishes to thank Prof. Xiao Wen for his support and encouragement.

\bibliographystyle{plain}
\bibliography{ref}

\begin{thebibliography}{10}

\bibitem{Antonelli2000}
Peter~L. Antonelli.
\newblock {\em Equivalence {P}roblem for {S}ystems of {S}econd {O}rder
  {O}rdinary {D}ifferential {E}quations}.
\newblock Encyclopedia of Mathematics, Kluwer Academic, Dordrecht, 2000.

\bibitem{PLA2003}
Peter~L. Antonelli, editor.
\newblock {\em Handbook of {F}insler {G}eometry}.
\newblock Kluwer Academic, Dordrecht, 1 edition, 2003.

\bibitem{AIM1993}
Peter~L. Antonelli, Roman~S. Ingarden, and Makoto Matsumoto.
\newblock {\em The Theory of Sprays and Finsler Spaces with Applications in
  Physics and Biology}.
\newblock Springer, Dordrecht, 1993.

\bibitem{adrt2016}
Valery Antonov, Diana Doli\'canin, Valery~G. Romanovski, and J\'anos T\'oth.
\newblock Invariant planes and periodic oscillations in the {M}ay--{L}eonard
  asymmetric model.
\newblock {\em MATCH Commun. Math. Comput. Chem.}, 76:455--474, 2016.

\bibitem{BChernS2000}
David~D. Bao, Shiing~S. Chern, and Zhong~M. Shen.
\newblock {\em An {I}ntroduction to {R}iemann--{F}insler {G}eometry}.
\newblock Springer, Berlin, 2000.

\bibitem{BHS2012}
Christian~G. B{\"o}ehmer, Tiberiu Harko, and Sorin~V. Sabau.
\newblock Jacobi stability analysis of dynamical systems -- applications in
  gravitation and cosmology.
\newblock {\em Adv. Theor. Math. Phys.}, 16:1145--1196, 2012.

\bibitem{bhlr2015}
Fran\c{c}ois Boulier, Mao~A. Han, Fran\c{c}ois Lemaire, and Valery~G.
  Romanovski.
\newblock Qualitative investigation of a gene model using computer algebra
  algorithms.
\newblock {\em Program. Comput. Soft.}, 41:105--111, 2015.

\bibitem{bdetal:2020}
Russell Bradford, James~H. Davenport, Matthew England, Hassan Errami, Vladimir
  Gerdt, Dima Grigoriev, Charles Hoyt, Marek Ko\v{s}ta, Ovidiu Radulescu,
  Thomas Sturm, and Andreas Weber.
\newblock Identifying the parametric occurrence of multiple steady states for
  some biological networks.
\newblock {\em J. Symb. Comput.}, 98:84--119, 2020.

\bibitem{Cartan1933}
\'Elie Cartan.
\newblock Observations sur le m\'emoire pr\'ec\'edent.
\newblock {\em Math. Z.}, 37:619--622, 1933.

\bibitem{CCMYZ13}
Chang~B. Chen, Robert~M. Corless, Marc~Moreno Maza, Pei Yu, and Yi~M. Zhang.
\newblock An application of regular chain theory to the study of limit cycles.
\newblock {\em Int. J. Bifur. Chaos}, 23:1350154--1--21, 2013.

\bibitem{Chern1939}
Shiing-Shen Chern.
\newblock Sur la geometrie d'un systeme d'equations differentielles du second
  ordre.
\newblock {\em Bull. Sci. Math.}, 63:206--212, 1939.

\bibitem{CMS1996}
Mike Crampin, Eduardo Mart\'inez, and Willy Sarlet.
\newblock Linear connections for systems of second-order ordinary differential
  equations.
\newblock {\em Ann. Inst. Henri Poincare}, 65:223--249, 1996.

\bibitem{msed2007}
Mohab Safey~El Din.
\newblock Testing sign conditions on a multivariate polynomial and
  applications.
\newblock {\em Math. Comput. Sci.}, 1:177--207, 2007.

\bibitem{Douglas1928}
Jesse Douglas.
\newblock The general geometry of paths.
\newblock {\em Ann of Math.}, 29:143--169, 1928.

\bibitem{GY2017}
Manish~K. Gupta and Chiranjeev~K. Yadav.
\newblock Jacobi stability analysis of r{\"o}ssler system.
\newblock {\em Int. J. Bifur. Chaos}, 27:1750056--1--13, 2017.

\bibitem{HGR10}
Mao~A. Han and Valery~G. Romanovski.
\newblock Estimating the number of limit cycles in polynomials systems.
\newblock {\em J. Math. Anal. Appl.}, 368:491--497, 2010.

\bibitem{HHLY2015}
Tiberiu Harko, Chor~Y. Ho, Chun~S. Leung, and Stan Yip.
\newblock Jacobi stability analysis of the lorenz system.
\newblock {\em Int. J. Geom. Methods Mod. Phys.}, 12:1550081--1--23, 2015.

\bibitem{HPPS2016}
Tiberiu Harko, Praiboon Pantaragphong, and Sorin~V. Sabau.
\newblock Kosambi--cartan--chern (kcc) theory for higher-order dynamical
  systems.
\newblock {\em Int. J. Geom. Methods Mod. Phys.}, 13:1650014--1--24, 2016.

\bibitem{HLS97}
Hoon Hong, Richard Liska, and Stanly Steinberg.
\newblock Testing stability by quantifier elimination.
\newblock {\em J. Symb. Comput.}, 24:161--187, 1997.

\bibitem{HTX15}
Hoon Hong, Xiao~X. Tang, and Bi~C. Xia.
\newblock Special algorithm for stability analysis of multistable biological
  regulatory systems.
\newblock {\em J. Symb. Comput.}, 70:112--135, 2015.

\bibitem{HNW2022}
Bo~Huang, Wei Niu, and Dong~M. Wang.
\newblock Symbolic computation for the qualitative theory of differential
  equations.
\newblock {\em Acta Math. Sci.}, 42B:2478--2504, 2022.

\bibitem{ISSAC2024HWY}
Bo~Huang, Dong~M. Wang, and Jing Yang.
\newblock Jacobi stability analysis for systems of odes using symbolic
  computation.
\newblock In H.~Jonathan, L.~Wen-Shin, and S.~Chen, editors, {\em Proceedings
  of the 2024 International Symposium on Symbolic and Algebraic Computation},
  pages 180--187, NY, USA, 2024. ACM.

\bibitem{EKW00}
M'hammed~E. Kahoui and Andreas Weber.
\newblock Deciding hopf bifurcations by quantifier elimination in a
  software-component architecture.
\newblock {\em J. Symb. Comput.}, 30:161--179, 2000.

\bibitem{MSK1929}
Morris~S. Knebelman.
\newblock Colineations and motions in generalized spaces.
\newblock {\em Amer. J. Math.}, 51:527--564, 1929.

\bibitem{Kosambi1933}
Damodar~D. Kosambi.
\newblock Parallelism and path-space.
\newblock {\em Math. Z.}, 37:608--818, 1933.

\bibitem{Olga1997}
Olga Krupkov\'a.
\newblock {\em The Geometry of Ordinary Variational Equations}.
\newblock Springer, Heidelberg, 1997.

\bibitem{Brad1999}
Brad Lackey.
\newblock A model of trophodynamics.
\newblock {\em Nonlinear Anal.}, 35:37--57, 1999.

\bibitem{Matthew2016}
Matthew~J. Lake and Tiberiu Harko.
\newblock Dynamical behavior and jacobi stability analysis of wound strings.
\newblock {\em The European Physical Journal C}, 76(6):311, 2016.

\bibitem{PLMT1985}
Peter Lancaster and Miron Tismenetsky.
\newblock {\em The Theory of Matrices: With Applications}.
\newblock Academic Press, London, 1985.

\bibitem{DLFR}
Daniel Lazard and Fabrice Rouillier.
\newblock Solving parametric polynomial systems.
\newblock {\em J. Symb. Comput.}, 42:636--667, 2007.

\bibitem{Librescu2003}
Liviu Librescu, Gianfranco Chiocchia, and Piergiovanni Marzocca.
\newblock Implications of cubic physical/aerodynamic non-linearities on the
  character of the flutter instability boundary.
\newblock {\em International Journal of Non-Linear Mechanics}, 38(2):173--199,
  2003.

\bibitem{acj17}
Adam Mahdi, Claudio Pessoa, and Jonathan~D. Hauenstein.
\newblock A hybrid symbolic-numerical approach to the center-focus problem.
\newblock {\em J. Symb. Comput.}, 82:57--73, 2017.

\bibitem{RMAM1982}
Richard~K. Miller and Anthony~N. Michel.
\newblock {\em Ordinary Differential Equations}.
\newblock Academic Press, New York London, 1982.

\bibitem{vd09}
Valery~G. Romanovski and Douglas~S. Shafer.
\newblock {\em The Center and Cyclicity Problems: A Computational Algebra
  Approach}.
\newblock Birkh{\"a}user, Boston, 2009.

\bibitem{SVS22005}
Sorin~V. Sabau.
\newblock Some remarks on {J}acobi stability.
\newblock {\em Nonlinear Anal.}, 63:143--153, 2005.

\bibitem{SVS12005}
Sorin~V. Sabau.
\newblock Systems biology and deviation curvature tensor.
\newblock {\em Nonlinear Anal. Real World Appl.}, 6:563--587, 2005.

\bibitem{Shen2001}
Zhong~M. Shen.
\newblock {\em Differential Geometry of Spray and Finsler Spaces}.
\newblock Springer, Dordrecht, 2001.

\bibitem{swek-mcs2009}
Thomas Sturm, Andreas Weber, Essam~O. Abdel-Rahman, and M’hammed~E. Kahoui.
\newblock Investigating algebraic and logical algorithms to solve hopf
  bifurcation problems in algebraic biology.
\newblock {\em Math. Comput. Sci.}, 2:493--515, 2009.

\bibitem{Synge1926}
John~L. Synge.
\newblock On the geometry of dynamics.
\newblock {\em Phil. Trans. Royal Soc., A}, 226:31--106, 1926.

\bibitem{TP1999}
Virendra~K. Tewari and Niranjan Prasad.
\newblock Three-dof modelling of tractor seat-operator system.
\newblock {\em Journal of Terramechanics}, 36:207--219, 1999.

\bibitem{dw91}
Dong~M. Wang.
\newblock Mechanical manipulation for a class of differential systems.
\newblock {\em J. Symb. Comput.}, 12:233--254, 1991.

\bibitem{WDMXBC2005}
Dong~M. Wang and Bi~C. Xia.
\newblock Stability analysis of biological systems with real solution
  classification.
\newblock In {\em Proceedings of the 2005 International Symposium on Symbolic
  and Algebraic Computation}, pages 354--361, 2005.

\bibitem{YHX01}
Lu~Yang, Xiao~R. Hou, and Bi~C. Xia.
\newblock A complete algorithm for automated discovering of a class of
  inequality-type theorems.
\newblock {\em Sci. China (Ser. F)}, 44:33--49, 2001.

\bibitem{LYBX05}
Lu~Yang and Bi~C. Xia.
\newblock Real solution classifications of parametric semi-algebraic systems.
\newblock pages 281--289, Herstellung und Verlag, Norderstedt, 2005. In:
  Algorithmic Algebra and Logic -- Proceedings of the A3L 2005 (A. Dolzmann, A.
  Seidl, and T. Sturm, eds.).

\bibitem{Zheng2008}
Guo~Y. Zheng and Yi~R. Yang.
\newblock Chaotic motions and limit cycle flutter of two-dimensional wing in
  supersonic flow.
\newblock {\em Acta Mechanica Solida Sinica}, 21(5):441--448, 2008.

\end{thebibliography}

%
\appendix
\newpage

\section{Dynamic Behaviors of the Model of Wound Strings}\label{sec-A}

We provide the trajectories of system \eqref{WS-2} near the four fixed points $(x_1,x_2)=(\pm 2,\pm 1)$ and $(x_1,x_2)=(\pm 2,\mp 1)$ in Figure \ref{fig:ex-WS-2}, and the deviation curve of $\xi_1(t)$ and $\xi_2(t)$ near the fixed point $(x_1,x_2)=(-2,1)$ in Figure \ref{fig:ex-WS-3}. 


\begin{figure}[t]
\begin{minipage}{0.49\linewidth}
    \centering
    \includegraphics[width=0.9\linewidth]{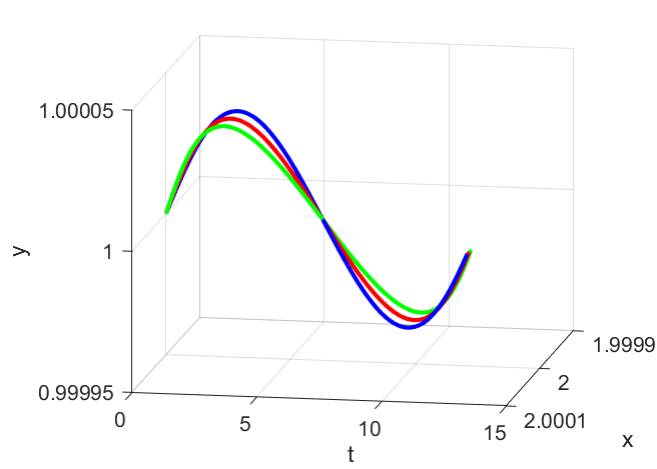}
    \caption*{(a) Trajectories of system \eqref{WS-2} starting from the fixed point $(x_1,x_2)=(2,1)$.}
\end{minipage}
\begin{minipage}{0.49\linewidth}
    \centering
    \includegraphics[width=0.9\linewidth]{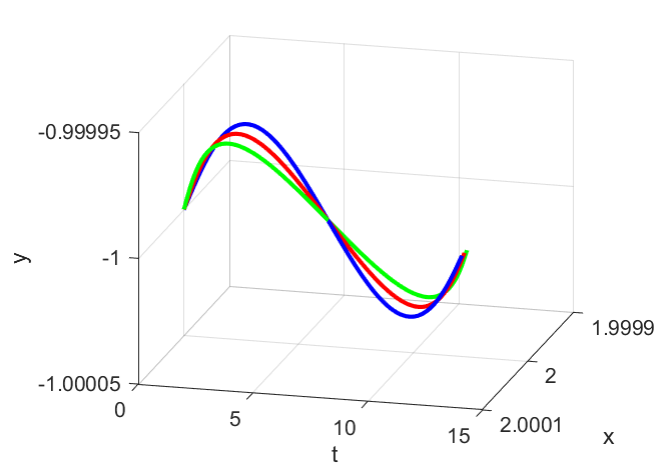}
    \caption*{(b) Trajectories of system \eqref{WS-2} starting from the fixed point $(x_1,x_2)=(2,-1)$.}
\end{minipage}
\begin{minipage}{0.49\linewidth}
    \centering
    \includegraphics[width=0.9\linewidth]{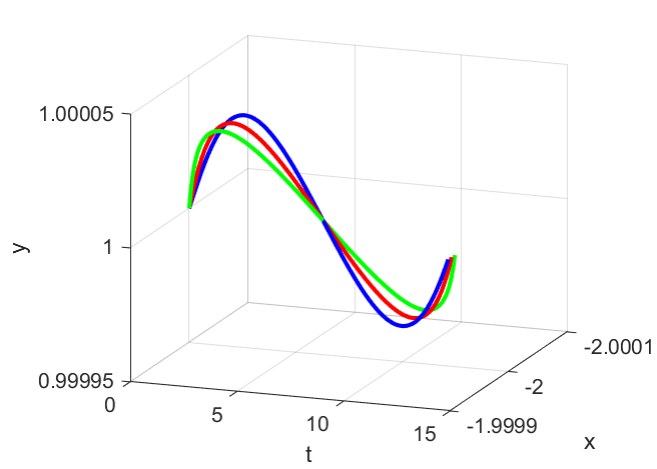}
    \caption*{(c) Trajectories of system \eqref{WS-2} starting from the fixed point $(x_1,x_2)=(-2,1)$.}
\end{minipage}
\begin{minipage}{0.49\linewidth}
    \centering
    \includegraphics[width=0.9\linewidth]{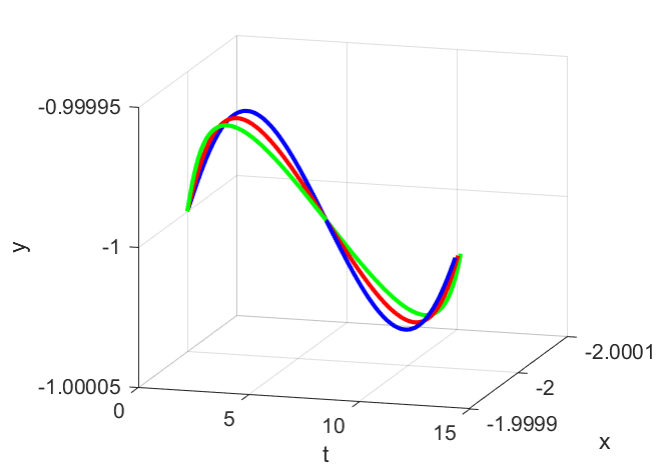}
    \caption*{(d) Trajectories of system \eqref{WS-2} starting from the fixed point $(x_1,x_2)=(-2,-1)$.}
\end{minipage}
\caption{Trajectories of system \eqref{WS-2} starting from the fixed point $(x_1,x_2)=(\pm 2,\pm 1)$ and $(x_1,x_2)=(\pm 2,\mp 1)$. The initial conditions of the blue curve, red curve, and green curve are $(x(0),y(0),\dot{x}(0),\dot{y}(0))=(x_1,x_2,10^{-5},2\times-10^{-5})$, $(x(0),y(0),\dot{x}(0),\dot{y}(0))=(x_1,x_2,2\times10^{-5},2\times 10^{-5})$, and $(x(0),y(0),\dot{x}(0),\dot{y}(0))=(x_1,x_2,3\times 10^{-5},2\times 10^{-5})$, respectively.} \label{fig:ex-WS-2}
\end{figure}

\begin{figure}[h]
\centering
\begin{minipage}{0.45\linewidth}
    \centering
    \includegraphics[width=0.99\linewidth]{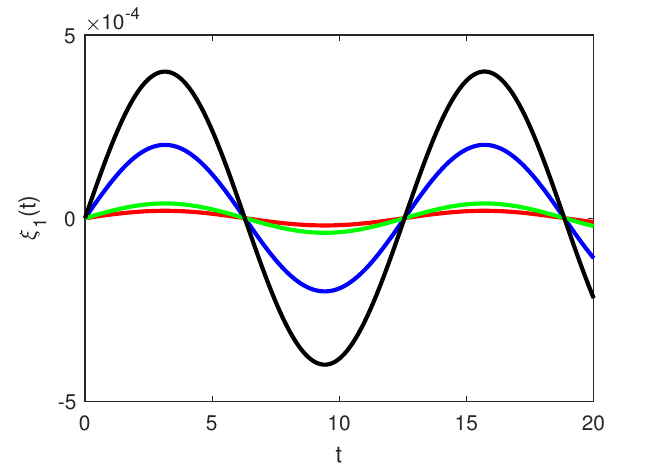}
    \caption*{(a) Deviation curve of $\xi_1(t)$.}
\end{minipage}
\begin{minipage}{0.45\linewidth}
    \centering
    \includegraphics[width=0.99\linewidth]{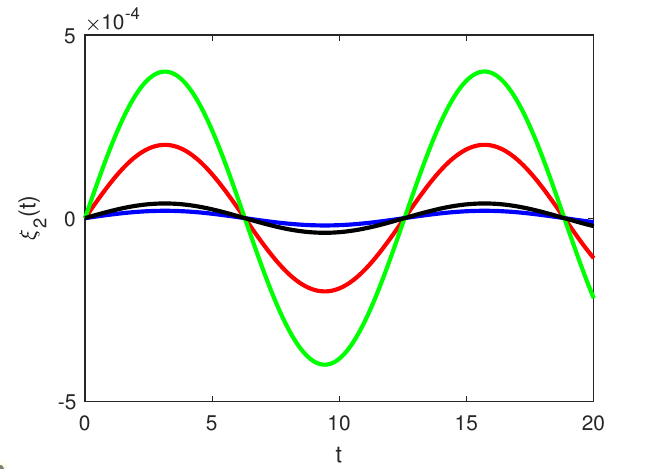}
    \caption*{(b) Deviation curve of $\xi_2(t)$.}
\end{minipage}
\caption{Deviation curve $(\xi_1(t),\xi_2(t))$ of system \eqref{WS-dev-1} near the fixed point $(x_1,x_2)=(-2,1)$. The initial conditions of the red curve, blue curve, green curve, and black curve are $(\xi_{10},\xi_{20})=(10^{-5},10^{-4})$, $(\xi_{10},\xi_{20})=(10^{-4},10^{-5})$, $(\xi_{10},\xi_{20})=(2\times 10^{-5},2\times 10^{-4})$, and $(\xi_{10},\xi_{20})=(2\times 10^{-4},2\times 10^{-5})$, respectively.}
\label{fig:ex-WS-3}
\end{figure}

\section{Dynamic Behaviors of the Airfoil Model}\label{sec-B}

\subsection{Case 1}
We take the parameter values $(M_\infty,V)=\left(\frac{2017}{256},\frac{83}{4}\right)\in\mathcal{C}_5$. In this case system \eqref{AM-3} has three fixed points $(\bar{x}_1,\bar{x}_2)\approx(\pm 0.1550,\mp 0.1202)$ and $(\bar{x}_1,\bar{x}_2)=(0,0)$. The system \eqref{AM-3} is Jacobi stable near the fixed points $(\bar{x}_1,\bar{x}_2)\approx(\pm 0.1550,\mp 0.1202)$, while is Jacobi unstable near the fixed point $(\bar{x}_1,\bar{x}_2)=(0,0)$.

The system \eqref{AM-3} has same deviation equations near the fixed points $(\bar{x}_1,\bar{x}_2)\approx(0.1550,-0.1202)$ and $(\bar{x}_1,\bar{x}_2)\approx(-0.1550,0.1202)$. We provide the graph of $||\boldsymbol{\xi}(t)||^2$ and $t^2$ near $t\approx0^+$ of the fixed point $(\bar{x}_1,\bar{x}_2)\approx(0.1550,-0.1202)$ in Figure \ref{fig:ex-AM-1}, the trajectories of system \eqref{AM-3} near the fixed points $(\bar{x}_1,\bar{x}_2)\approx(\pm 0.1550,\mp 0.1202)$ in Figure \ref{fig:ex-AM-2}, and the deviation curve of $\xi_1(t)$ and $\xi_2(t)$ near the fixed point $(\bar{x}_1,\bar{x}_2)\approx(0.1550,-0.1202)$ in Figure \ref{fig:ex-AM-3}. 

\begin{figure}[h]
    \centering
    \includegraphics[width=0.6\linewidth]{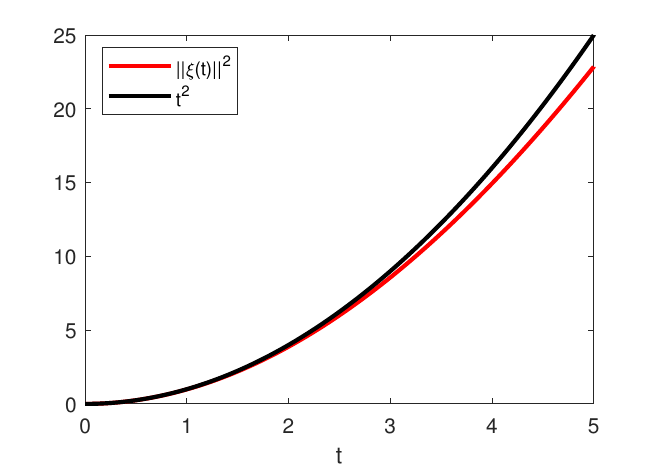}
    \caption{The curve graph of $t^2$ and $||\boldsymbol{\xi}(t)||^2$ of the corresponding deviation equations associated to system \eqref{AM-3}. The red curve is $||\boldsymbol{\xi}(t)||^2$ with the initial condition $(\xi_{10},\xi_{20}=(10^{-4},10^{-5})$.}
    \label{fig:ex-AM-1}
\end{figure}

\begin{figure}[h]
    \begin{minipage}{0.49\linewidth}
    \centering
    \includegraphics[width=0.9\linewidth]{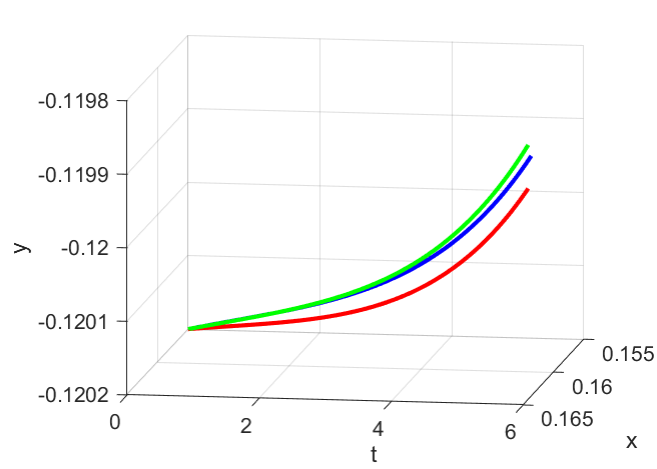}
    \caption*{(a) Trajectories of system \eqref{AM-3} starting from the fixed point $(0.1550,-0.1202)$.}
\end{minipage}
\begin{minipage}{0.49\linewidth}
    \centering
    \includegraphics[width=0.9\linewidth]{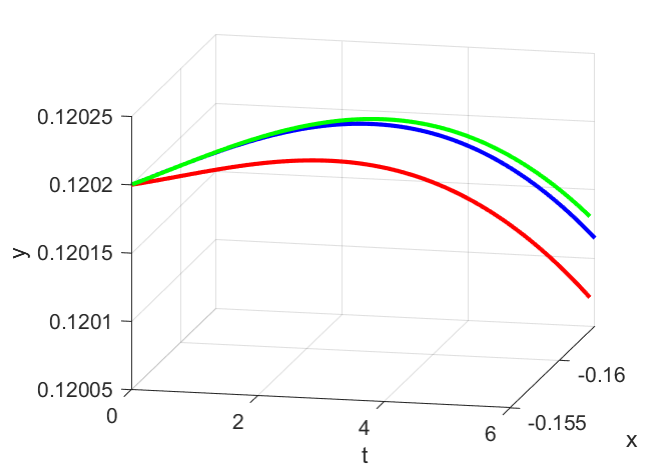}
    \caption*{(b) Trajectories of system \eqref{AM-3} starting from the fixed point $(-0.1550,0.1202)$.}
\end{minipage}

\caption{Trajectories of system \eqref{AM-3} starting from the fixed point $(\bar{x}_1,\bar{x}_2)\approx(\pm 0.1550,\mp 0.1202)$. The initial conditions of the blue curve, red curve, and green curve are $(x_1(0),x_2(0),y_1(0),y_2(0))=(\pm 0.1550,\mp 0.1202,10^{-5},2\times 10^{-5})$, $(x_1(0),x_2(0),y_1(0),y_2(0))=(\pm 0.1550,\mp 0.1202,10\times 10^{-5},10^{-5})$, and  $(x_1(0),x_2(0),y_1(0),y_2(0))=(\pm 0.1550,\mp 0.1202,10\times 10^{-5},2\times 10^{-5})$, respectively.} \label{fig:ex-AM-2}
\end{figure}

\begin{figure}[h]
\begin{minipage}{0.49\linewidth}
    \centering
    \includegraphics[width=0.99\linewidth]{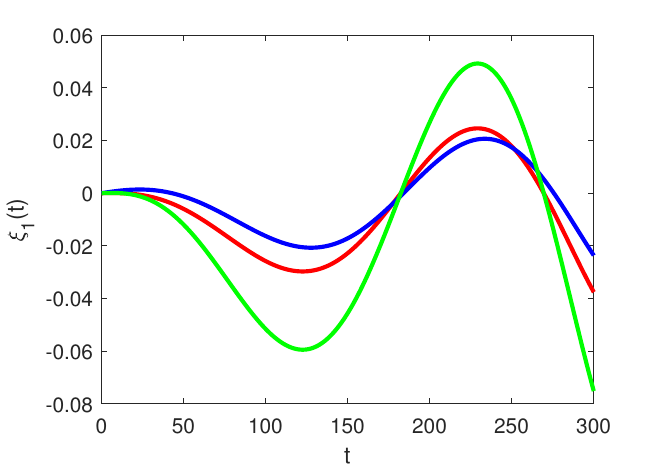}
    \caption*{(a) Deviation curve of $\xi_1(t)$.}
\end{minipage}
\begin{minipage}{0.49\linewidth}
    \centering
    \includegraphics[width=0.99\linewidth]{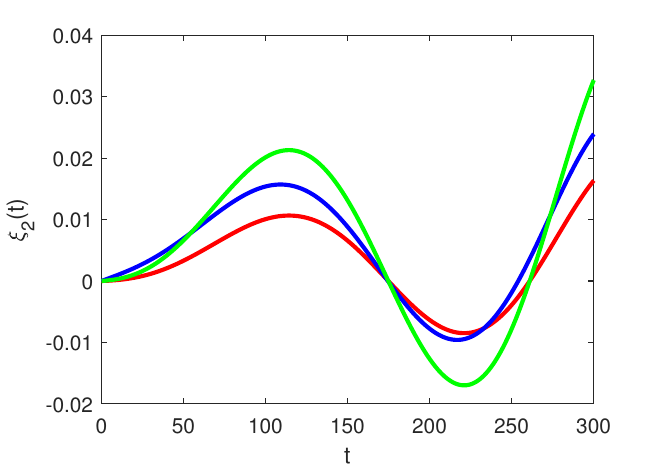}
    \caption*{(b) Deviation curve of $\xi_2(t)$.}
\end{minipage}
\caption{Deviation curve $(\xi_1(t),\xi_2(t))$ of the corresponding deviation equations associated to system \eqref{AM-3} near the fixed point $(\bar{x}_1,\bar{x}_2)\approx(0.1550,-0.1202)$. The initial conditions of the red curve, blue curve, and green curve are $(\xi_{10},\xi_{20})=(10^{-5},10^{-4})$, $(\xi_{10},\xi_{20})=(10^{-4},10^{-5})$, and $(\xi_{10},\xi_{20})=(2\times 10^{-5},2\times 10^{-4})$, respectively.}
\label{fig:ex-AM-3}
\end{figure}

We provide the graph of $||\boldsymbol{\xi}(t)||^2$ and $t^2$ near $t\approx0^+$ of the fixed point $(\bar{x}_1,\bar{x}_2)=(0,0)$ in Figure \ref{fig:ex-AM-4}, the trajectories of system \eqref{AM-3} near the fixed point $(\bar{x}_1,\bar{x}_2)=(0,0)$ in Figure \ref{fig:ex-AM-5}, and the deviation curve of $\xi_1(t)$ and $\xi_2(t)$ near the fixed point $(\bar{x}_1,\bar{x}_2)=(0,0)$ in Figure \ref{fig:ex-AM-6}. 

\begin{figure}[h]
    \centering
    \includegraphics[width=0.6\linewidth]{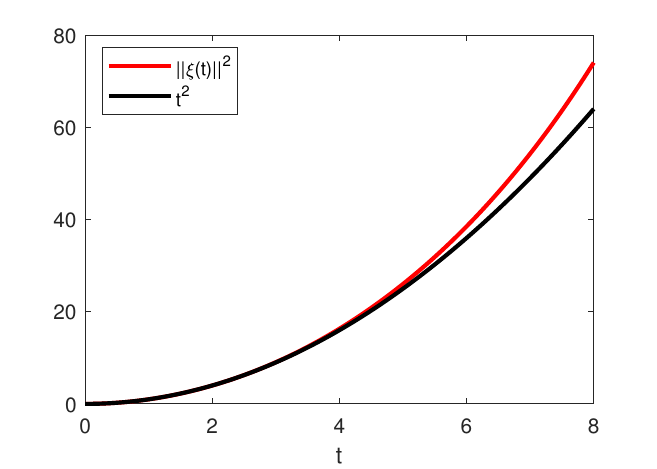}
    \caption{The curve graph of $t^2$ and $||\boldsymbol{\xi}(t)||^2$ of the corresponding deviation equations associated to system \eqref{AM-3}. The red curve is $||\boldsymbol{\xi}(t)||^2$ with the initial condition $(\xi_{10},\xi_{20}=(10^{-4},10^{-5})$.}
    \label{fig:ex-AM-4}
\end{figure}

\begin{figure}[h]
    \centering
    \includegraphics[width=0.6\linewidth]{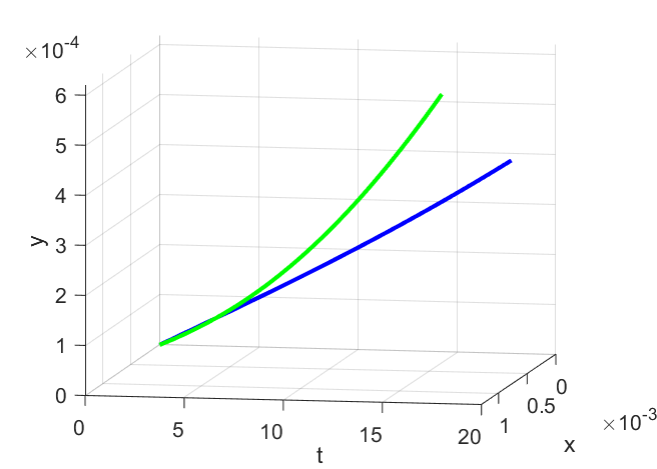}
\caption{Trajectories of system \eqref{AM-3} starting from the fixed point $(\bar{x}_1,\bar{x}_2)=(0,0)$. The initial conditions of the blue curve, and green curve are $(x_1(0),x_2(0),y_1(0),y_2(0))=(0,0,10^{-5},2\times 10^{-5})$ and  $(x_1(0),x_2(0),y_1(0),y_2(0))=(0,0,10\times 10^{-5},2\times 10^{-5})$, respectively.} \label{fig:ex-AM-5}
\end{figure}

\begin{figure}[h]
\begin{minipage}{0.49\linewidth}
    \centering
    \includegraphics[width=0.99\linewidth]{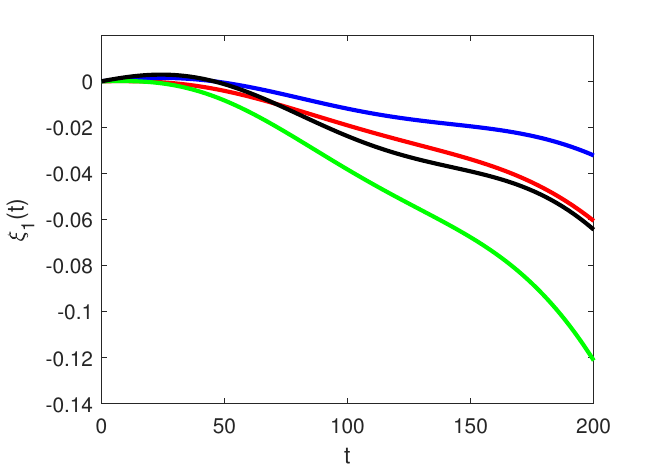}
    \caption*{(a) Deviation curve of $\xi_1(t)$.}
\end{minipage}
\begin{minipage}{0.49\linewidth}
    \centering
    \includegraphics[width=0.99\linewidth]{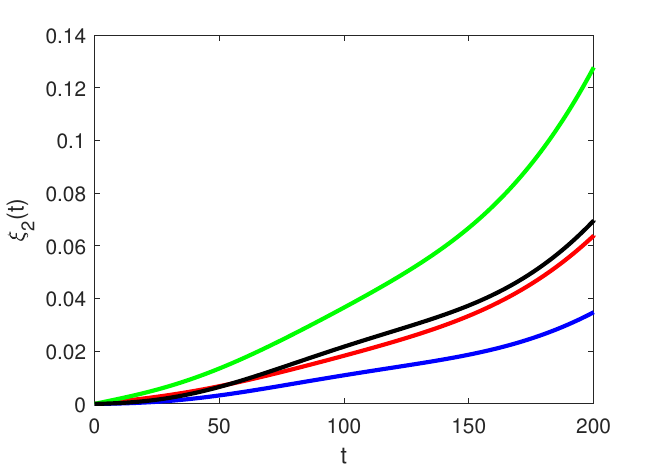}
    \caption*{(b) Deviation curve of $\xi_2(t)$.}
\end{minipage}
\caption{Deviation curve $(\xi_1(t),\xi_2(t))$ of the corresponding deviation equations associated to system \eqref{AM-3} near the fixed point $(\bar{x}_1,\bar{x}_2)=(0,0)$. The initial conditions of the red curve, blue curve, green curve, and black curve are $(\xi_{10},\xi_{20})=(10^{-5},10^{-4})$, $(\xi_{10},\xi_{20})=(10^{-4},10^{-5})$, $(\xi_{10},\xi_{20})=(2\times 10^{-5},2\times 10^{-4})$, and $(\xi_{10},\xi_{20})=(2\times 10^{-4},2\times 10^{-5})$, respectively.}
\label{fig:ex-AM-6}
\end{figure}

\subsection{Case 2}
We take the parameter values $(M_\infty,V)=\left(\frac{71}{16384},\frac{3}{16}\right)\in\mathcal{C}_2$. In this case system \eqref{AM-3} has one unique fixed point $(\bar{x}_1,\bar{x}_2)=(0,0)$. The system \eqref{AM-3} is Jacobi stable near the fixed point $(\bar{x}_1,\bar{x}_2)=(0,0)$. We provide the graph of $||\boldsymbol{\xi}(t)||^2$ and $t^2$ near $t\approx0^+$ of the fixed point $(\bar{x}_1,\bar{x}_2)=(0,0)$ in Figure \ref{fig:ex-AM-7}, the trajectories of system \eqref{AM-3} near the fixed point $(\bar{x}_1,\bar{x}_2)=(0,0)$ in Figure \ref{fig:ex-AM-8}, and the deviation curve of $\xi_1(t)$ and $\xi_2(t)$ near the fixed point $(\bar{x}_1,\bar{x}_2)=(0,0)$ in Figure \ref{fig:ex-AM-9}. 

\begin{figure}[h]
    \centering
    \includegraphics[width=0.6\linewidth]{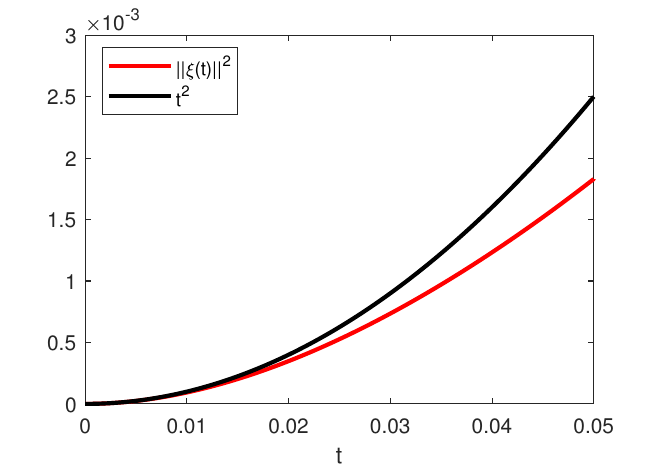}
    \caption{The curve graph of $t^2$ and $||\boldsymbol{\xi}(t)||^2$ of the corresponding deviation equations associated to system \eqref{AM-3}. The red curve is $||\boldsymbol{\xi}(t)||^2$ with the initial condition $(\xi_{10},\xi_{20}=(10^{-4},10^{-5})$.}
    \label{fig:ex-AM-7}
\end{figure}

\begin{figure}[h]
    \centering
    \includegraphics[width=0.6\linewidth]{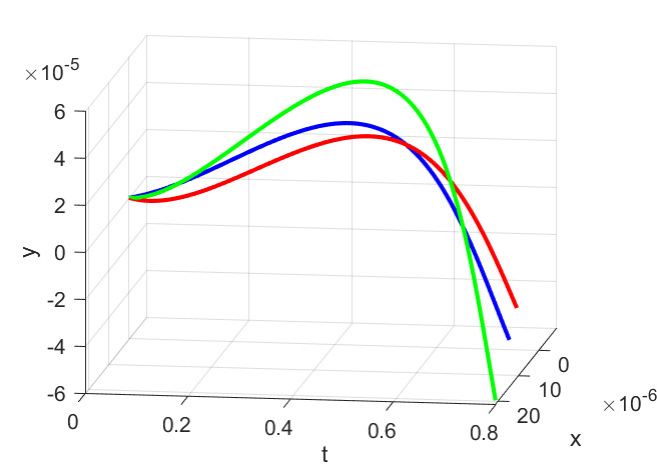}
\caption{Trajectories of system \eqref{AM-2} starting from the fixed point $(\bar{x}_1,\bar{x}_2)=(0,0)$. The initial conditions of the blue curve, red curve, and green curve are $(x_1(0),x_2(0),y_1(0),y_2(0))=(0,0,10^{-5},2\times 10^{-5})$, $(x_1(0),x_2(0),y_1(0),y_2(0))=(0,0,10\times 10^{-5},10^{-5})$, and  $(x_1(0),x_2(0),y_1(0),y_2(0))=(0,0,10\times 10^{-5},2\times 10^{-5})$, respectively.} \label{fig:ex-AM-8}
\end{figure}

\begin{figure}[h]
\begin{minipage}{0.49\linewidth}
    \centering
    \includegraphics[width=0.99\linewidth]{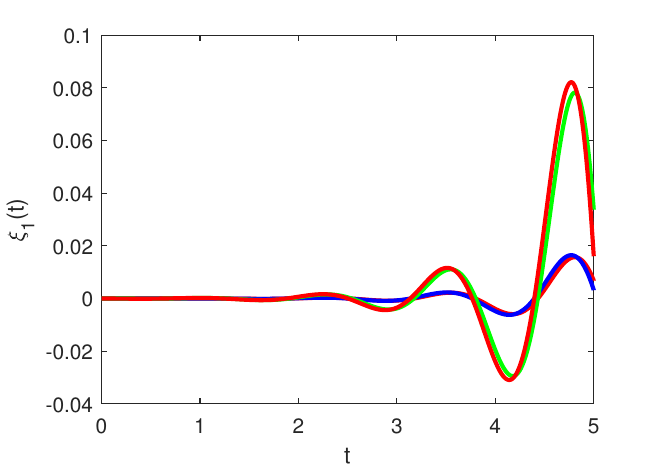}
    \caption*{(a) Deviation curve of $\xi_1(t)$.}
\end{minipage}
\begin{minipage}{0.49\linewidth}
    \centering
    \includegraphics[width=0.99\linewidth]{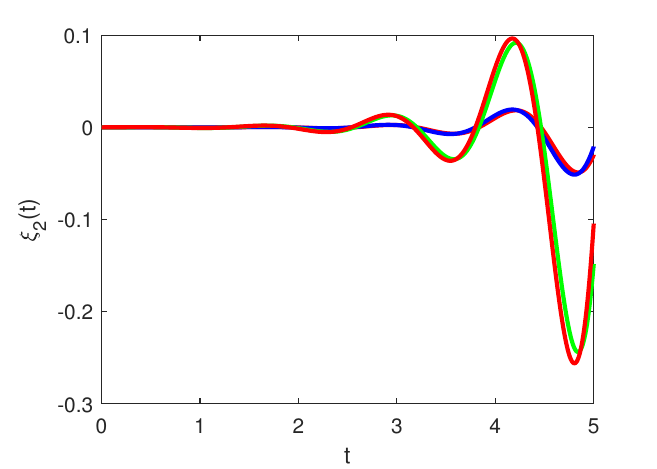}
    \caption*{(b) Deviation curve of $\xi_2(t)$.}
\end{minipage}
\caption{Deviation curve $(\xi_1(t),\xi_2(t))$ of the corresponding
deviation equations associated to system \eqref{AM-3} near the fixed point $(\bar{x}_1,\bar{x}_2)=(0,0)$. The initial conditions of the red curve, blue curve, green curve, and black curve are $(\xi_{10},\xi_{20})=(10^{-5},10^{-4})$, $(\xi_{10},\xi_{20})=(10^{-4},10^{-5})$, $(\xi_{10},\xi_{20})=(2\times 10^{-5},2\times 10^{-4})$, and $(\xi_{10},\xi_{20})=(2\times 10^{-4},2\times 10^{-5})$, respectively.}
\label{fig:ex-AM-9}
\end{figure}

\section{A 3-DOF Tractor Seat-Operator Model.}\label{sec-C}

The model equations are
\begin{equation}\label{3D}
\begin{split}
M_1\ddot{z}_1&+C_1\dot{z}_1+K_1z_1-C_2(\dot{z}_2-\dot{z}_1)-K_2(z_2-z_1)=0,\\
M_2\ddot{z}_2&+C_2(\dot{z}_2-\dot{z}_1)+K_2(z_2-z_1)-C_3(\dot{z}_3-\dot{z}_2)-K_3(z_3-z_2)=0,\\
M_3\ddot{z}_3&-C_3(\dot{z}_3-\dot{z}_2)-K_3(z_3-z_2)=0,
\end{split}
\end{equation}
where $z_1$, $z_2$, $z_3$ denote the displacement of seat mass, of operator pelvis mass, and of thorax and head mass (m). The symbols $M_1$, $M_2$, $M_3$ denote the seat mass, operator pelvis mass, and operator thorax and head mass (kg). The coefficients $K_1$,$K_2$, $K_3$ are spring constant of suspension spring, of cushion material, and of pelvis-thorax connection (N/m). The symbols $C_1$, $C_2$, $C_3$ denote the damping coefficient of suspension damper, of cushion material, and of pelvis-thorax connection (N-s/m).

\begin{table}[h]
\caption{Model parameters }\label{Tab-A}
\begin{center}
\begin{tabular}{cccccccc}
\toprule
  Cases & $M_1$ & $M_2$ & $M_3$ & $K_1$ & $K_2$ & $C_1$ & $C_2$\\ \hline
   1 & $\frac{31}{5}$ & $65\times\frac{5}{7}$ & $65\times\frac{2}{7}$ & $22600$ & $37730$ & $920$ & $159$\\ \hline
   2 & $\frac{31}{5}$ & $65\times\frac{5}{7}$ & $65\times\frac{2}{7}$ & $15000$ & $37730$ & $750$ & $159$\\ \hline
   3 & $\frac{31}{5}$ & $65\times\frac{5}{7}$ & $65\times\frac{2}{7}$ & $25000$ & $37730$ & $750$ & $159$\\ \hline
   4 & $\frac{31}{5}$ & $65\times\frac{5}{7}$ & $65\times\frac{2}{7}$ & $20000$ & $37730$ & $500$ & $159$\\ \hline
   5 & $\frac{31}{5}$ & $65\times\frac{5}{7}$ & $65\times\frac{2}{7}$ & $20000$ & $37730$ & $750$ & $159$\\ \hline
   6 & $\frac{31}{5}$ & $65\times\frac{5}{7}$ & $65\times\frac{2}{7}$ & $20000$ & $37730$ & $1000$ & $159$\\ \hline
   7 & $\frac{31}{5}$ & $36$ & $14$ & $20000$ & $37730$ & $750$ & $159$\\ \hline
   8 & $\frac{31}{5}$ & $46$ & $19$ & $20000$ & $37730$ & $750$ & $159$\\ \hline
   9 & $\frac{31}{5}$ & $57$ & $23$ & $20000$ & $37730$ & $750$ & $159$\\
   \bottomrule
\end{tabular}
\end{center}
\end{table}

The model requires cushion and suspension properties as the input variables. Different types of suspension seats and cushion types were tested in the laboratory under controlled conditions following the SAE recommended test procedures. Specifications of these seats are given in Table \ref{Tab-A} (see \cite{TP1999}). Choosing the model parameters suggested in Table \ref{Tab-A}, we obtain the following result.

\begin{theorem}\label{main-kcc-seat}
The model \eqref{3D} cannot have Jacobi stable fixed point under each of the nine cases in Table \ref{Tab-A}.
\end{theorem}

To write system \eqref{3D} in the KCC standard form, we change the notation as
\begin{equation*}
x=(x_1,x_2,x_3)=(z_1,z_2,z_3),\quad y=(y_1,y_2,y_3)=(\dot{z}_1,\dot{z}_2,\dot{z}_3).
\end{equation*}
Then system \eqref{3D} becomes:
\begin{equation}\label{3D-1}
\begin{split}
\ddot{x}_1+2G^1(x,y)=0,\quad
\ddot{x}_2+2G^2(x,y)=0,\quad 
\ddot{x}_3+2G^3(x,y)=0,
\end{split}
\end{equation}
where
\begin{equation}\label{3D-2}
\begin{split}
G^1&=\frac{1}{2}\left(\frac{K_1+K_2}{M_1}x_1-\frac{K_2}{M_1}x_2+\frac{C_1+C_2}{M_1}y_1-\frac{C_2}{M_1}y_2\right),\\
G^2&=\frac{1}{2}\left(-\frac{K_2}{M_2}x_1+\frac{K_2+K_3}{M_2}x_2-\frac{K_3}{M_2}x_3-\frac{C_2}{M_2}y_1+\frac{C_2+C_3}{M_2}y_2-\frac{C_3}{M_2}y_3\right),\\
G^3&=\frac{1}{2}\left(\frac{K_3}{M_3}x_2-\frac{K_3}{M_3}x_3+\frac{C_3}{M_3}y_2-\frac{C_3}{M_3}y_3\right).\nonumber
\end{split}
\end{equation}


Applying Algorithm 2 to system \eqref{3D-1}, we obtain the deviation equations:
\begin{equation}\label{3D-dev-1}
    \begin{split}
        \frac{d^2 \xi_1}{dt^2}&=-\frac{K_1+K_2}{M_1}\xi_1+\frac{K_2}{M_1}\xi_2-\frac{C_1+C_2}{M_1}\frac{d\xi_1}{dt}+\frac{C_2}{M_1}\frac{d\xi_2}{dt},\\
        \frac{d^2 \xi_2}{dt^2}&=\frac{K_2}{M_2}\xi_1-\frac{K_2+K_3}{M_2}\xi_2+\frac{K_3}{M_2}\xi_3+\frac{C_2}{M_2}\frac{d\xi_1}{dt}-\frac{C_2+C_3}{M_2}\frac{d\xi_2}{dt}+\frac{C_3}{M_2}\frac{d\xi_3}{dt},\\
        \frac{d^2 \xi_3}{dt^2}&=-\frac{K_3}{M_3}\xi_2+\frac{K_3}{M_3}\xi_3-\frac{C_3}{M_3}\frac{d\xi_1}{dt}+\frac{C_3}{M_3}\frac{d\xi_2}{dt}.
    \end{split}
\end{equation}


We take the parameter values $(M_1,M_2,M_3,K_1,K_2,K_3,C_1,C_2,C_3)=(\frac{31}{5},57,23,20000,37730,\\1000,750,159,1000)$. In this case system \eqref{3D-1} has a unique fixed point $(\bar{x}_1,\bar{x}_2,\bar{x}_3)=(0,0,0)$. We provide the graph of $||\boldsymbol{\xi}(t)||^2$ and $t^2$ near $t\approx0^+$ in Figure \ref{fig:ex-3D-1}, the trajectories of system \eqref{3D-1} near the fixed point $(\bar{x}_1,\bar{x}_2,\bar{x}_3)=(0,0,0)$ in Figure \ref{fig:ex-3D-2}, and the deviation curve of $\xi_1(t)$ and $\xi_2(t)$ near the fixed point $(\bar{x}_1,\bar{x}_2,\bar{x}_3)=(0,0,0)$ in Figure \ref{fig:ex-3D-3}. It follows from Figure \ref{fig:ex-3D-1} that $||\boldsymbol{\xi}(t)||^2 > t^2$ near $t\approx0^+$. According to KCC theory, this indicates that the fixed point is Jacobi unstable. 
\begin{figure}[h]
    \centering
    \includegraphics[width=0.6\linewidth]{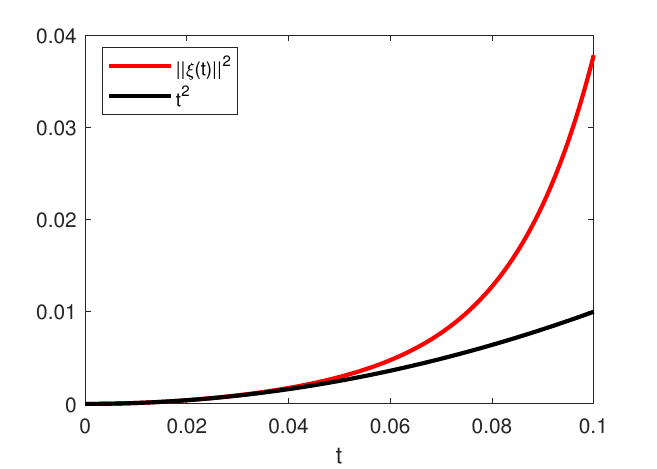}
    \caption{The curve graph of $t^2$ and $||\boldsymbol{\xi}(t)||^2$ of system \eqref{3D-dev-1}. The red curve is $||\boldsymbol{\xi}(t)||^2$ with the initial condition $(\xi_{10},\xi_{20},\xi_{30})=(10^{-5},10^{-4},10^{-4})$.}
    \label{fig:ex-3D-1}
\end{figure}

\begin{figure}[h]
    \begin{minipage}{0.33\linewidth}
    \centering
    \includegraphics[width=0.9\linewidth]{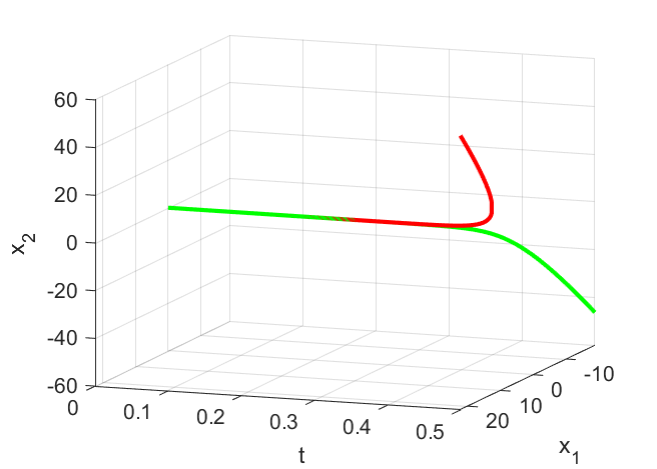}
    \caption*{(a) Projection of system \eqref{3D-1} in the $x_1-x_2$ direction.}
\end{minipage}
\begin{minipage}{0.33\linewidth}
    \centering
    \includegraphics[width=0.9\linewidth]{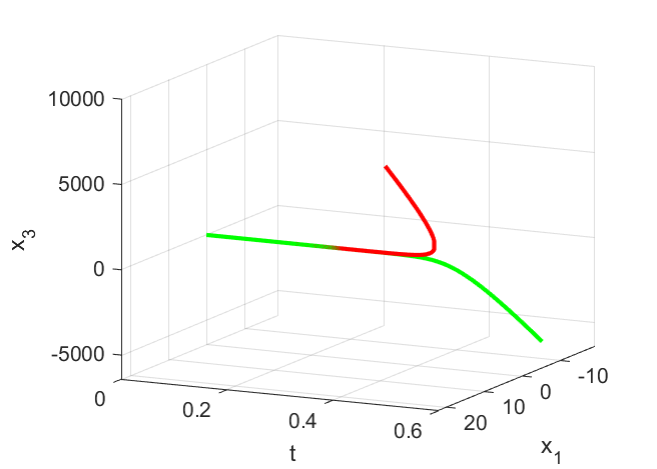}
    \caption*{(b) Projection of system \eqref{3D-1} in the $x_1-x_3$ direction.}
\end{minipage}
\begin{minipage}{0.33\linewidth}
    \centering
    \includegraphics[width=0.9\linewidth]{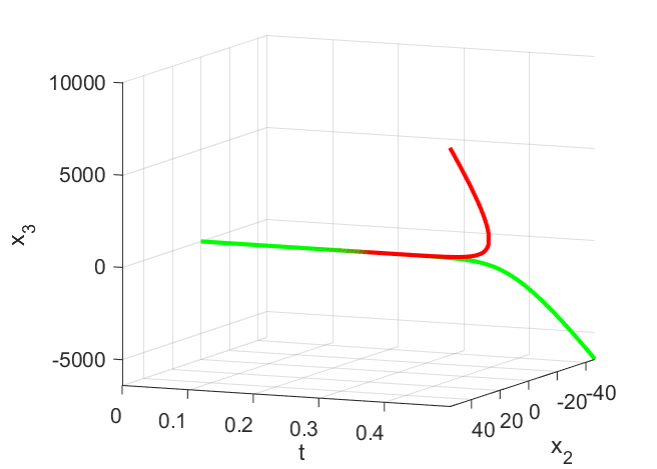}
    \caption*{(c) Projection of system \eqref{3D-1} in the $x_2-x_3$ direction.}
\end{minipage}
\caption{Trajectories of system \eqref{3D-1} starting from the fixed point $(\bar{x}_1,\bar{x}_2,\bar{x}_3)=(0,0,0)$. The initial conditions of the red curve and green curve are  $(x_1(0),x_2(0),x_3(0),y_1(0),y_2(0),y_3(0))=(0,0,0,10^{-4},10^{-5},10^{-4})$ and $(x_1(0),x_2(0),x_3(0),y_1(0),y_2(0),y_3(0))=(0,0,0,10^{-4},10^{-4},10^{-5})$, respectively.} \label{fig:ex-3D-2}
\end{figure}

\begin{figure}[h]
\begin{minipage}{0.33\linewidth}
    \centering
    \includegraphics[width=0.99\linewidth]{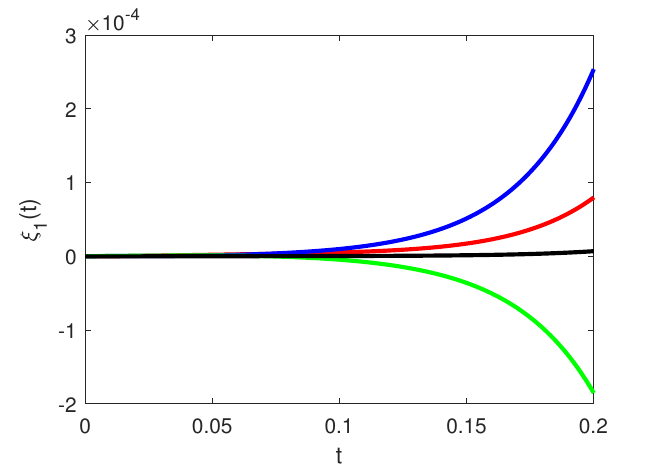}
    \caption*{(a) Deviation curve of $\xi_1(t)$.}
\end{minipage}
\begin{minipage}{0.33\linewidth}
    \centering
    \includegraphics[width=0.99\linewidth]{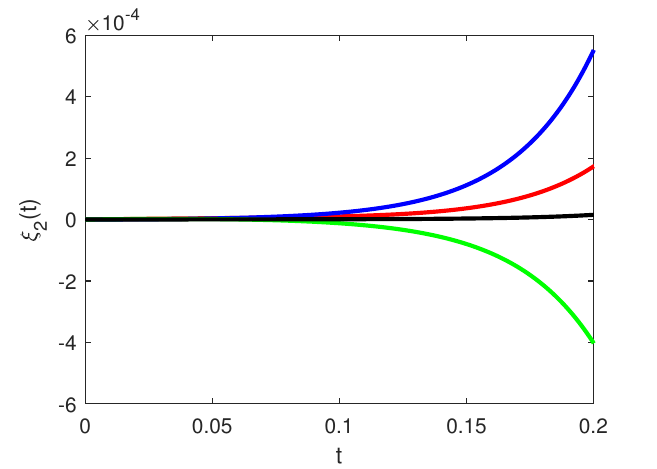}
    \caption*{(b) Deviation curve of $\xi_2(t)$.}
\end{minipage}
\begin{minipage}{0.33\linewidth}
    \centering
    \includegraphics[width=0.99\linewidth]{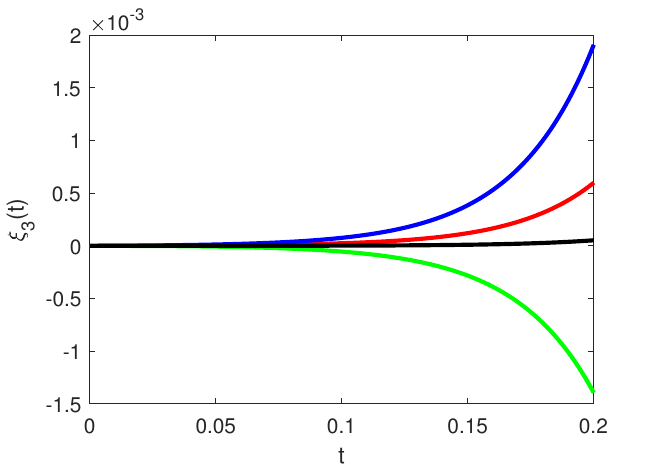}
    \caption*{(c) Deviation curve of $\xi_3(t)$.}
\end{minipage}
\caption{Deviation curve $(\xi_1(t),\xi_2(t),\xi_3(t))$ of system \eqref{3D-dev-1} near the fixed point $(x_1,x_2,x_3)=(0,0,0)$. The initial conditions of the red curve, blue curve, green curve, and black curve are $(\xi_{10},\xi_{20},\xi_{30})=(10^{-5},10^{-4},10^{-4})$, $(\xi_{10},\xi_{20},\xi_{30})=(10^{-4},10^{-5},10^{-4})$, $(\xi_{10},\xi_{20},\xi_{30})=(10^{-4},10^{-4},10^{-5})$, and $(\xi_{10},\xi_{20},\xi_{30})=(10^{-5},10^{-5},10^{-5})$, respectively.}
\label{fig:ex-3D-3}
\end{figure}

\end{document}